\documentclass{article}

\usepackage{amsmath}
\usepackage{amsfonts}
\usepackage{amssymb}
\usepackage{amsthm,color}
\usepackage{amsmath,amsthm,amssymb}
\usepackage{graphicx,mathtools}
\usepackage{cite}
\usepackage{enumitem}
\usepackage{lipsum}
\usepackage{algorithm}      
\usepackage{algorithmic}    
\usepackage{dsfont}
\usepackage{hyperref}

\newtheorem{Theorem}{Theorem}[section]

\newtheorem{Lemma}[Theorem]{Lemma}
\newtheorem{Corollary}[Theorem]{Corollary}

\def\({\left(}
\def\){\right)}

\usepackage{subcaption} 

\graphicspath{{pics/}}

\begin{document}

\title{Factorization method for a clamped obstacle from near-field measurements via a far-field transformation}
\author{General Ozochiawaeze\footnote{Department of Mathematics, Purdue University, West Lafayette, IN 47907, USA, \texttt{oozochia@purdue.edu}} \, and Isaac Harris\footnote{Department of Mathematics, Purdue University, West Lafayette, IN 47907, USA, \texttt{harri814@purdue.edu}}}
\date{}
\maketitle
\begin{abstract}
This paper considers an inverse shape problem for recovering an unknown
impenetrable clamped obstacle in two dimensions from near-field point
source and dipole measurements for the biharmonic Helmholtz equation in
the frequency domain. The measured data consist of the scattered field
and its normal derivative on a closed measurement curve surrounding the
obstacle. Since the associated near-field operator does not directly admit the symmetric factorization required by the factorization method, we introduce a far-field transformation. This transformation is defined independently of the obstacle and augments the near-field operator into the associated far-field operator which does admit a symmetric factorization. This yields a rigorous complete characterization of the obstacle by the factorization method theory and leads to a practical reconstruction algorithm based on the spectral data of the transformed operator. Numerical experiments are presented to demonstrate
the effectiveness of the proposed method, with synthetic near-field data
generated using the method of fundamental solutions. We also consider reconstructions using only scattered-field measurements generated by point sources, demonstrating the potential for reduced the amount of measured data.
\end{abstract}

\section{Introduction}
Obstacle scattering problems constitute a fundamental class of problems in scattering theory. The direct problem concerns determining the scattered field generated by a known obstacle and a prescribed incident field, whereas the inverse problem seeks to recover the unknown obstacle from measurements of the scattered field. This paper considers the inverse shape problem of recovering a clamped obstacle in a thin elastic plate from measured multi-static near-field data. The resulting two-dimensional scattering problem is modeled by a frequency-domain biharmonic `Helmholtz' equation coupled with clamped boundary conditions. In contrast to acoustic and electromagnetic scattering, flexural (i.e. biharmonic)  scattering in a thin elastic plate describes out-of-plane displacements and are governed by a fourth-order equation that reflects the bending behavior of the plate. This higher-order structure gives rise to distinct analytical features and presents additional challenges for the mathematical analysis and numerical reconstruction of inverse scattering problems.

The study of biharmonic scattering problems is motivated by a broad range of applications. These include the design of ultra-broadband elastic cloaks for vibration control in vehicles and earthquake-resistant passive systems for smart buildings (see \cite{Farhat1}, \cite{Farhat2}, \cite{Farhat3}). Platonic crystals, consisting of periodic arrays of cavities, provide a mechanism for manipulating flexural waves analogous to photonic and phononic crystals \cite{Farhat0}. Acoustic black holes can passively trap flexural waves, with applications in noise reduction, energy harvesting, and biomedical devices \cite{Pelat2020}. Biharmonic models also arise in nondestructive testing and structural health monitoring in the aerospace industry (see e.g., \cite{Wang2022} and \cite{Zhou2022}), as well as in the modeling of sea ice and ice shelves (see \cite{Askham2025}). These applications have recently motivated considerable interest in the analysis of biharmonic wave propagation and scattering.

In this paper, we develop a factorization method for reconstructing an unknown clamped obstacle from near-field measurements generated by a point-source incident field. Numerical methods for inverse scattering can be broadly classified into quantitative and qualitative approaches. Quantitative methods, like the domain derivative \cite{Fink2024} and the continuation method \cite{Bao2007}, typically formulate the reconstruction as an optimization problem and iteratively update the unknown boundary based on the measured data. Qualitative methods, on the other hand, are direct imaging approaches that seek to identify the unknown scatterer by evaluating an appropriate imaging functional constructed from the scattering data, without explicitly solving an optimization problem for the boundary. Such methods are often computationally attractive and require relatively limited a priori information about the unknown obstacle. The factorization method, originally proposed by Kirsch in \cite{Kirsch1998} for the inverse acoustic scattering problem, belongs to this class of qualitative (i.e. non-iterative inversion) methods. Compared with the linear sampling method \cite{ColtonKirsch1996}, the factorization method provides a sharper characterization of sampling points according to whether they lie inside or outside the scatterer and, moreover, admits a necessary and sufficient criterion that provides a rigorous foundation for its mathematical analysis.

Considerable attention has been devoted to the development of numerical methods for inverse scattering problems governed by the biharmonic wave equation. For clamped obstacles, an optimization-based reconstruction method using near-field measurements of both the scattered field and its Laplacian was developed in \cite{ChangGuo2023}, extending an earlier decomposition approach of Colton and Kirsch \cite{ColtonKress2019}. In \cite{BourgeoisRecoquillay2020}, a near-field linear sampling method was developed for reconstructing obstacles with clamped and free plate boundary conditions. For far-field measurements, in \cite{Guo2024} the linear sampling method was developed for identifying clamped obstacles, while \cite{HarrisLiOzochiawaeze2026} developed a factorization of the far-field operator that removes a restrictive assumption concerning Dirichlet eigenvalues and proposed an extended sampling method for recovering clamped obstacles with one or few incident directions. In addition to these approaches, direct sampling methods have been adapted to biharmonic scattering in \cite{HarrisLeeLi2025, ZhuGe2025}.

Factorization methods have also been developed for inverse scattering problems governed by the biharmonic wave equation, including applications to penetrable absorbing media \cite{CejaAyalaHarrisOzochiawaeze2026} and impenetrable clamped and free plate obstacles with far-field data \cite{Zhu2026}. Furthermore, more closely related to the present work, a factorization method was developed for recovering simply supported obstacles with Poisson's ratio equal to one from near-field point-source measurements \cite{HarrisKleefeld2026}. In that work, by splitting the biharmonic scattered field into propagating and evanescent components, the authors reformulated the near-field data in terms of an acoustic near-field operator and subsequently transformed it into an acoustic far-field operator. However, extending this approach to clamped obstacles is considerably more challenging due to the more complicated structure of the associated near-field operator. In particular, uniqueness results obtained in \cite{BourgeoisRecoquillay2020} for recovering both clamped and free plate obstacles require richer near-field data consisting of the scattered fields generated by both point-source and dipole incident fields, together with their normal derivatives. Consequently, the propagating and evanescent decomposition alone does not yield an analogous acoustic near-field operator for the clamped problem, making the development of a factorization method considerably more delicate. In the present work, we overcome the corresponding difficulty for near-field data by establishing a novel factorization of the corresponding near-field operator and introducing explicit bounded operators (defined independently of the scatterer) that augment the biharmonic near–field operator into the associated far–field operator for the clamped obstacle, to which the results in \cite{Zhu2026} can be applied.

The paper is organized as follows. In Sect. \ref{intro}, we formulate the direct and inverse shape scattering problems, explicitly identifying the need for multi-static near-field measurements generated by point sources and dipoles to ensure unique recovery. In Sect. \ref{NF_operator_sect}, we introduce the near-field operator for the clamped obstacle and derive a novel factorization of this operator. In Sect. \ref{transform_sect}, we derive and precompute the far-field transformation using separation of variables, which augments the near-field operator into a far-field operator for a clamped obstacle and leads to a factorization method. In Sect. \ref{numerics_sect}, we describe the numerical implementation, where the near-field data are computed using the method of fundamental solutions as in \cite{karageorghislesnic2024}, and present numerical reconstructions. Finally, in Sect. \ref{conclude_sect}, we summarize our results and discuss possible directions for future work.

\section{The Direct and Inverse Scattering Problem}\label{intro}
In this section, we introduce both the direct and inverse scattering problem under consideration. Let $D \subset \mathbb{R}^2$ be a bounded obstacle with a $C^2$-smooth boundary $\partial D$, and let $\mathcal{C}$ be a measurement curve surrounding $D$. We illuminate the obstacle by prescribing an incident wave denoted
\[
{u^{i}(\cdot \, , y), \quad \text{where }y\in \mathcal C\text{ such that }\overline D\subset \text{Int}(\mathcal C)}
\]
for some smooth closed curve $\mathcal C$. Here, $\mathcal C$ is the measurement curve where the sources and receivers are assumed to be placed. {For our problem, we consider two types of incident fields: a point-source incident wave, and a dipole incident wave given by its normal derivative.} To this end, we define the point-source incident wave by 
\begin{align}
\mathbb G(x,y)\coloneqq -\frac{\mathrm{i}}{8\kappa^2}\left(H_0^{(1)}(\kappa|x-y|)-H_0^{(1)}(\mathrm i\kappa|x-y|)\right),\quad\text{for } x\neq y
\end{align}
with $H_0^{(1)}$ defined as the Hankel function of the first kind of order zero and $\kappa>0$ is the wavenumber. This is the radiating fundamental solution for the operator $\Delta^2-\kappa^4$ such that $\mathbb G(\cdot \, ,y)$ satisfies
\[
\Delta^2\mathbb G(\cdot \, ,y)-\kappa^4\mathbb G(\cdot \, ,y)=-\delta(\cdot-y)\quad \text{in }\mathbb R^2
\]
along with radiation conditions
\[
\lim_{r\to\infty}\sqrt r(\partial_r \mathbb G(x,y)- \text{i}\kappa\mathbb G(x,y))=0\quad \text{and}\quad \lim_{r\to\infty}\sqrt r(\partial_r \Delta\mathbb  G(x,y)- \text{i}\kappa\Delta \mathbb G(x,y))=0
\]
with $r=|x|$. The above limits are assumed to hold uniformly for $\hat x=x/r$ and $y\in \mathcal C$. Note that we have that $\mathbb G(x,y)$ is  $C^{\infty}$-smooth for all $x\neq y$ in $\mathbb R^2$.

For the clamped obstacle, the total displacement field and its normal derivative have zero trace on the boundary $\partial D$.
For a given incident field located at the source $y \in \mathcal{C}$, the corresponding scattered field $u^{s}(\cdot , y)\in H_{\text{loc}}^2(\mathbb R^2\setminus\overline D)$ satisfies the biharmonic scattering problem for a fixed wave number $\kappa>0$ given by
\begin{align}\label{eqnbcs}
    \Delta^2u^s-\kappa^4u^s=0,\quad x\in\mathbb R^2\setminus\overline D
\end{align}
with the clamped boundary conditions
\begin{align}\label{bcs}
    u^s=-u^i\quad \text{and}\quad \partial_\nu u^s=-\partial_\nu u^i,\quad x\in\partial D.
\end{align}
In addition, to close the system we assume that the scattered field satisfies the Sommerfeld radiation conditions
\begin{align}\label{SRCs}
    \lim_{r\to\infty} \sqrt{r}(\partial_r u^s-\mathrm{i}\kappa u^s)=0\quad \text{and}\quad \lim_{r\to\infty} \sqrt{r}(\partial_r \Delta u^s-\mathrm{i}\kappa \Delta u^s)=0, \quad r=|x|
\end{align}
which is assumed to hold uniformly in $\hat x=x/r$. We note that the well-posedness of the forward problem \eqref{eqnbcs}-\eqref{SRCs} was shown to hold in \cite{BourgeoisHazard2020}. 

{For the inverse problem, we consider sources and receivers distributed on the measurement curve $\mathcal C$. For each source point $y \in \mathcal C$, we can probe the obstacle with the incident field
\[
u^{i}(\cdot \, , y) = \mathbb{G}(\cdot \, , y), \,\,\text{where we denote the corresponding scattered field by } u^{s}(\cdot \, , y),
\]
and for the auxiliary dipole incident field
\[
\tilde{u}^{i}(\cdot \, , y) = \partial_{\nu_y}\mathbb{G}(\cdot \, , y), \,\, \text{where we denote the corresponding scattered field by } \tilde{u}^{s}(\cdot \, , y).
\]
Here $\nu_y$ denotes the outward unit normal at $y \in \mathcal C$.} 
For every source point $y \in \mathcal C$, we assume that we can measure both scattered fields together with their normal derivatives
%
These measurements form the \emph{multistatic} near-field data set
\begin{equation}\label{NF_dataset}
\mathbb M :=
\Big\{
u^{s}(x,y), \ \tilde u^{s}(x,y), \ \partial_{\nu_x}u^{s}(x,y), \ \partial_{\nu_x}\tilde u^{s}(x,y)
:\, \text{for all} \, \, x,y \in \mathcal C
\Big\}.
\end{equation}
The inverse problem is to reconstruct the unknown obstacle $D$ from the knowledge of $\mathbb M$. 
The uniqueness of this inverse problem has been established in Theorem 2.3 of \cite{BourgeoisRecoquillay2020}, which shows that the measured multistatic near-field data uniquely determines the impenetrable clamped obstacle.

\section{Factorization of near-field operator}\label{NF_operator_sect}
In this section, we derive a factorization of the near-field operator. For our analysis of the inverse problem, we proceed to define the (multistatic) near-field operator $\mathcal N : [L^2(\mathcal C)]^2\to [L^2(\mathcal C)]^2$ by
\begin{equation}\label{NF_operator}
\mathcal N
\begin{pmatrix}
h\\
t
\end{pmatrix}(x)
=
\begin{pmatrix}
\displaystyle
\int_{\mathcal C}
\big(u^{s}(x,y)\,h(y) + \tilde u^{s}(x,y)\,t(y)\big)\,ds(y)
\\[2ex]
\displaystyle
\partial_{\nu_x}
\int_{\mathcal C}
\big(u^{s}(x,y)\,h(y) + \tilde u^{s}(x,y)\,t(y)\big)\,ds(y)
\end{pmatrix},
\qquad x \in \mathcal C.
\end{equation}
Note that the near-field operator $\mathcal{N}$ associated with the data set $\mathbb M$ was first studied in \cite{BourgeoisRecoquillay2020}. We build upon their foundational analysis and develop a new factorization of the near-field operator. 

To begin, we now introduce the following auxiliary operators. We first consider the obstacle-to-data operator
$\mathcal B : H^{-3/2}(\partial D)\times H^{-1/2}(\partial D)
\to [L^2(\mathcal C)]^2$ defined by
\begin{equation}
\mathcal B
\begin{pmatrix}
\tau\\
\sigma
\end{pmatrix}(x)
=
\begin{pmatrix}
\displaystyle
\int_{\partial D}
\Big(
\mathbb G(x,y)\tau(y)
+
\partial_{\nu_y}\mathbb G(x,y)\sigma(y)
\Big)\,ds(y)
\\[2ex]
\displaystyle
\partial_{\nu_x}
\int_{\partial D}
\Big(
\mathbb G(x,y)\tau(y)
+
\partial_{\nu_y}\mathbb G(x,y)\sigma(y)
\Big)\,ds(y)
\end{pmatrix},
\qquad x \in \mathcal C.
\end{equation}
Moreover, we define the solution operator
$\mathcal U : H^{3/2}(\partial D)\times H^{1/2}(\partial D)
\to [L^2(\mathcal C)]^2$
defined by
\begin{equation}\label{soln_operator}
\mathcal U(\phi, \psi)^\top = \left(v, \, \partial_{\nu} v\right)^\top\Big|_{\mathcal C},
\end{equation}
where $v \in H^2_{\mathrm{loc}}(\mathbb R^2\setminus\overline D)$ is the radiating solution of the exterior problem
\begin{equation}\label{ext_BVP}
\Delta^2 v - \kappa^4 v = 0 
\qquad \text{in } \mathbb{R}^2 \setminus \overline{D},
\end{equation}
satisfying the clamped boundary conditions
\begin{align}\label{clamped_bcs_inhomog}
v = \phi \quad\text{and}\quad \partial_{\nu} v = \psi 
\qquad \text{on } \partial D.
\end{align}
for any $(\phi,\psi)^\top\in H^{3/2}(\partial D)\times H^{1/2}(\partial D)$. 

To obtain a factorization of the near-field operator, we first determine the transpose of the obstacle-to-data operator $\mathcal B$. 
We denote by $\langle \cdot, \cdot \rangle_{[L^2(\mathcal C)]^2}$ the natural bilinear pairing (or integral pairing) between $[L^2(\mathcal C)]^2$ and itself, and by $\langle \cdot, \cdot \rangle$ the duality pairing between the Sobolev spaces $H^{-3/2}(\partial D)\times H^{-1/2}(\partial D)$ and $H^{3/2}(\partial D)\times H^{1/2}(\partial D)$. Using this bilinear form, we compute
\begin{align*}
\left\langle
\mathcal B
\begin{pmatrix}
\tau\\
\sigma
\end{pmatrix},
\begin{pmatrix}
h\\
t
\end{pmatrix}
\right\rangle_{[L^2(\mathcal C)]^2}
&=
\int_{\mathcal C}
\int_{\partial D}
\Big(
\mathbb G(x,y)\tau(y)
+
\partial_{\nu_y}\mathbb G(x,y)\sigma(y)
\Big)\,ds(y)\,h(x)\,ds(x)
\\
&\quad+
\int_{\mathcal C}
\partial_{\nu_x}
\left(
\int_{\partial D}
\Big(
\mathbb G(x,y)\tau(y)
+
\partial_{\nu_y}\mathbb G(x,y)\sigma(y)
\Big)\,ds(y)
\right)
t(x)\,ds(x)
\\[1ex]
&=
\int_{\partial D}\tau(y)
\int_{\mathcal C}
\Big(
\mathbb G(x,y)h(x)
+
\partial_{\nu_x}\mathbb G(x,y)t(x)
\Big)\,ds(x)\,ds(y)
\\
&\quad+
\int_{\partial D}\sigma(y)
\int_{\mathcal C}
\Big(
\partial_{\nu_y}\mathbb G(x,y)h(x)
+
\partial_{\nu_y}\partial_{\nu_x}\mathbb G(x,y)t(x)
\Big)\,ds(x)\,ds(y)
\\[1ex]
&=
\int_{\partial D}
\tau(y)
\int_{\mathcal C}
\Big(
\mathbb G(y,x)h(x)
+
\partial_{\nu_x}\mathbb G(y,x)t(x)
\Big)\,ds(x)\,ds(y)
\\
&\quad+
\int_{\partial D}
\sigma(y)
\int_{\mathcal C}
\Big(
\partial_{\nu_y}\mathbb G(y,x)h(x)
+
\partial_{\nu_y}\partial_{\nu_x}\mathbb G(y,x)t(x)
\Big)\,ds(x)\,ds(y)
\\[1ex]
&=
\left\langle
\begin{pmatrix}
\tau\\
\sigma
\end{pmatrix},
\mathcal B^\top
\begin{pmatrix}
h\\
t
\end{pmatrix}
\right\rangle.
\end{align*}
Because of the symmetry of the fundamental solution, i.e.,  $\mathbb G(x,y)=\mathbb G(y,x)$ and a change of variables above, the transpose operator $\mathcal B^\top : [L^2(\mathcal{C})]^2 \to H^{3/2}(\partial D)\times H^{1/2}(\partial D)$ is given by
\[
\mathcal B^\top
\begin{pmatrix}
h\\
t
\end{pmatrix}(x)
=
\begin{pmatrix}
\displaystyle
\int_{\mathcal C}
\Big(
\mathbb G(x,y)h(y)
+
\partial_{\nu_y}\mathbb G(x,y)t(y)
\Big)\,ds(y)
\\[3ex]
\displaystyle
\int_{\mathcal C}
\Big(
\partial_{\nu_x}\mathbb G(x,y)h(y)
+
\partial_{\nu_x}\partial_{\nu_y}\mathbb G(x,y)t(y)
\Big)\,ds(y)
\end{pmatrix},
\qquad x\in\partial D.
\]
With these definitions, the near-field operator admits the following
factorization.
\begin{Lemma}
The near-field operator $\mathcal N$ admits the factorization
\[
\mathcal N=-\mathcal U\mathcal B^\top .
\]
\end{Lemma}
\begin{proof}
Let $(h,t)^\top\in [L^2(\mathcal C)]^2$ and define the incident field
\[
v^i(x)=
\int_{\mathcal C}
\big(
\mathbb G(x,y)h(y)
+\partial_{\nu_y}\mathbb G(x,y)t(y)
\big)\,ds(y).
\]
By the definition of $\mathcal B^\top$, the Cauchy data of $v^i$ on
$\partial D$ satisfy
\[
\begin{pmatrix}
v^i\\
\partial_\nu v^i
\end{pmatrix}
=
\mathcal B^\top
\begin{pmatrix}
h\\
t
\end{pmatrix}.
\]
Let $v$ denote the unique radiating solution to \eqref{ext_BVP}-\eqref{clamped_bcs_inhomog} with clamped boundary conditions given by 
\[
v=-v^i,\qquad
\partial_\nu v=-\partial_\nu v^i
\quad\text{on }\partial D .
\]
Therefore, by the definition of $\mathcal U$,
\[
\mathcal U\mathcal B^\top
\begin{pmatrix}
h\\
t
\end{pmatrix}
=
-
\begin{pmatrix}
v^s|_{\mathcal C}\\
\partial_\nu v^s|_{\mathcal C}
\end{pmatrix}.
\]
On the other hand, {by well-posedness of the forward problem \eqref{eqnbcs}-\eqref{SRCs} and superposition}, the scattered field $v$ admits the representation
\[
v(x)
=
\int_{\mathcal C}
\big(
u^s(x,y)h(y)+\tilde u^s(x,y)t(y)
\big)\,ds(y),
\]
where $u^s(\cdot,y)$ and $\tilde u^s(\cdot,y)$ are the scattered fields
generated by the point-source incident fields
$\mathbb G(\cdot,y)$ and $\partial_{\nu_y}\mathbb G(\cdot,y)$,
respectively. Hence, using the definition of $\mathcal N$ in
\eqref{NF_operator}, we obtain
\[
\mathcal U\mathcal B^\top
\begin{pmatrix}
h\\
t
\end{pmatrix}
=
-\mathcal N
\begin{pmatrix}
h\\
t
\end{pmatrix}
\quad \text{ which gives that } \quad \mathcal N=-\mathcal U\mathcal B^\top 
\]
proving the claim. 
\end{proof}
We now derive a symmetric factorization of the near-field operator that is similar to Remark 2 in \cite{BourgeoisRecoquillay2020}. To proceed, we introduce the boundary integral operator
\[
\mathcal S:
H^{-3/2}(\partial D)\times H^{-1/2}(\partial D)
\to
H^{3/2}(\partial D)\times H^{1/2}(\partial D),
\]
defined by
\begin{equation}\label{coercive_op}
\mathcal S
\begin{pmatrix}
\tau\\
\sigma
\end{pmatrix}
(x)
=
\begin{pmatrix}
\displaystyle \int_{\partial D}
\Big(
\mathbb G(x,y)\tau(y)
+
\partial_{\nu_y}\mathbb G(x,y)\sigma(y)
\Big)\,ds(y)
\\[1.5ex]
\displaystyle \partial_{\nu_x}\int_{\partial D}
\Big(
\mathbb G(x,y)\tau(y)
+
\partial_{\nu_y}\mathbb G(x,y)\sigma(y)
\Big)\,ds(y)
\end{pmatrix},
\qquad x\in \partial D,
\end{equation}
which is a bounded operator due to the potential theory of the biharmonic wave equation. With this, obtaining the relation
\begin{equation}\label{relation_A}
\mathcal B=\mathcal U\mathcal S
\end{equation}
is a simple consequence of the definition of the three operators. Additionally, we remark that [Proposition 3,\cite{BourgeoisRecoquillay2020}]
establishes the compactness, injectivity, and dense range properties of
$\mathcal N$, $\mathcal B$, and $\mathcal U$. Moreover, [Proposition 7,
\cite{BourgeoisRecoquillay2020}] shows that
$\mathcal S$ is an isomorphism provided that
$\kappa$ is not the fourth root of a Dirichlet eigenvalue of the
bilaplacian $\Delta^2$ in $D$.
Combining this relation with the factorization of the near-field operator
established prior, we obtain the following factorization of the near-field operator.
\begin{Corollary}\label{NF_factorization_thm}
If $\kappa$ is not the fourth root of a Dirichlet eigenvalue of the bilaplacian $\Delta^2$ in $D$, we obtain the classical factorization
\[
\mathcal N
=
-\,\mathcal B\,\mathcal S^{-1}\,\mathcal B^\top.
\]
\end{Corollary}
\begin{proof}
The assumption on $\kappa$ ensures that
$\mathcal S^{-1}$ exists and is bounded. By the
relation \eqref{relation_A}, we obtain
\[
\mathcal U=\mathcal B\mathcal S^{-1}.
\]
Substituting this identity into the factorization
\[
\mathcal N=-\mathcal U\mathcal B^\top \quad \text{implies that } \quad \mathcal N=-\mathcal B\mathcal S^{-1}\mathcal B^\top,
\]
which proves the result.
\end{proof}

\section{Factorization Method via the Far-Field Transform} \label{transform_sect}
Although the factorization in Corollary \ref{NF_factorization_thm} provides a useful representation of the near-field operator, it lacks the symmetry of the standard factorization due to the appearance of transpose operator rather than the adjoint. Therefore, it does not directly fit into the standard
framework of the factorization method. To overcome this issue, we
introduce a far-field transformation of the multistatic near-field
operator $\mathcal N$, which leads to a symmetric factorization suitable
for the factorization method.

Given Cauchy data $(f,h)^\top\in H^{3/2}(\mathcal C)\times H^{1/2}(\mathcal C)$,
let $w\in H^2_{\mathrm{loc}}(\mathbb R^2\setminus \overline{\rm{Int}(\mathcal C)})$ be the unique radiating solution of the exterior biharmonic problem
\begin{equation}
\Delta^2 w - \kappa^4 w = 0 \quad \text{in } \mathbb R^2\setminus \overline{\rm{Int}(\mathcal C)}
\end{equation}
with clamped boundary data
\begin{equation}
w= f,\qquad \partial_{\nu} w = h \quad \text{on } \mathcal C.
\end{equation}
We define the far-field transform
\begin{align}\label{FF_transform}
\mathcal Q: H^{3/2}(\mathcal C)\times H^{1/2}(\mathcal C) \to L^2(\mathbb{S}^1) \quad \text{given by} \quad (\mathcal Q (f,h)^{\top})(\hat x)
= w^{\infty}(\hat x),\quad \hat x\in\mathbb S^1,
\end{align}
where $\mathbb S^1$ denotes the unit `sphere' in $\mathbb R^2$. It is well-known that the far-field pattern $w^\infty$ exists and that $w$ satisfies the asymptotic expansion 
\[
w(x)=\frac{\mathrm e^{\mathrm i\pi/4}}{\sqrt{8\pi \kappa}}\frac{\mathrm e^{\mathrm i\kappa |x|}}{|x|^{1/2}}\left\{w^{\infty}(\hat x,d)+O(|x|^{-1})\right\},\quad \text{as}\quad |x|\to\infty,
\]
where $\hat x=x/|x|\in \mathbb S^1$. 

Recall that the obstacle-to-data operator is defined by
\[
\mathcal B(\tau,\sigma)^\top=(v^s|_{\mathcal C},\partial_\nu v^s|_{\mathcal C})^\top \quad \text{where}\quad v^s(x)
=
\int_{\partial D}
\Big(
\mathbb G(x,y)\tau(y)
+
\partial_{\nu_y}\mathbb G(x,y)\sigma(y)
\Big)\,ds(y).
\]
To connect the near-field representation with the corresponding far-field
pattern, we apply the far-field transform $\mathcal Q:
H^{3/2}(\mathcal C)\times H^{1/2}(\mathcal C)
\to
L^2(\mathbb S^1)$ introduced in \eqref{FF_transform}. Since $v^s$ is a radiating solution of the
biharmonic wave equation in $\mathbb R^2\setminus\overline D$, the
transformed data satisfy
\[
(\mathcal Q\mathcal B(\tau, \sigma)^\top)(\hat x)
=
v^\infty(\hat x),
\]
where $v^\infty$ denotes the far-field pattern of $v^s$.
It is well-known that the far-field pattern of the biharmonic fundamental solution is given by
\[
\mathbb G^\infty(\hat x,y)
=
-\frac{1}{2\kappa^2}
\mathrm e^{-\mathrm i\kappa\hat x\cdot y}.
\]
So, we obtain that 
\begin{align*}
v^\infty(\hat x)
=
\int_{\partial D}
\left(
\mathbb G^\infty(\hat x,y)\tau(y)
+
\partial_{\nu_y}\mathbb G^\infty(\hat x,y)\sigma(y)
\right)\,ds(y)
=
-\frac{1}{2\kappa^2}
\int_{\partial D}
\left(
\mathrm e^{-\mathrm i\kappa\hat x\cdot y}\tau(y)
+
\partial_{\nu_y}
\mathrm e^{-\mathrm i\kappa\hat x\cdot y}\sigma(y)
\right)\,ds(y).
\end{align*}
To proceed, we introduce the Herglotz operator defined in \cite{HarrisLiOzochiawaeze2026}. Namely, 
\begin{align}\label{herglotz_op}
\mathcal H:
L^2(\mathbb S^1)
\to
H^{3/2}(\partial D)\times H^{1/2}(\partial D)
\quad \text{is defined by}\quad
\mathcal H g
=(v_g |_{\partial D},\partial_\nu v_g|_{\partial D})^\top
\end{align}
where 
\[
v_g(x)\coloneqq \int_{\mathbb S^1}\mathrm e^{\mathrm i\kappa x\cdot d}g(d)\,ds(d),\quad \text{for }g\in L^2(\mathbb S^1)\quad \text{and}\quad  x\in \mathbb{R}^2
\]
is the Herglotz wave function.
Direct calculations gives that the adjoint $\mathcal H^{*}: H^{-3/2}(\partial D)\times H^{-1/2}(\partial D)\to L^2(\mathbb S^1)$ is given by
\[
\mathcal H^*
\begin{pmatrix}
\tau\\
\sigma
\end{pmatrix}
(\hat x)
=
\int_{\partial D}
\left(
\mathrm e^{-\mathrm i\kappa \hat x\cdot y}\tau(y)
+
\partial_{\nu_y}
\mathrm e^{-\mathrm i\kappa \hat x\cdot y}\sigma(y)
\right)\,ds(y).
\]
This implies
\[
\mathcal Q\mathcal B
=
-\frac{1}{2\kappa^2}
\mathcal H^*.
\]
By taking the transpose of this expression, we obtain the equality $\mathcal B^\top \mathcal Q^\top=-\displaystyle\frac{1}{2\kappa^2}(\mathcal H^{*})^\top$. 

We now compute the transpose of the adjoint of the Herglotz operator. By definition of the transpose operator,
\[
\left\langle
(\mathcal H^*)^\top g,
\begin{pmatrix}
\tau\\
\sigma
\end{pmatrix}
\right\rangle
=
\left\langle
g,
\mathcal H^*
\begin{pmatrix}
\tau\\
\sigma
\end{pmatrix}
\right\rangle_{L^2(\mathbb S^1)} .
\]
Substituting the expression for $\mathcal H^*$ yields
\begin{align*}
\left\langle
(\mathcal H^*)^\top g,
\begin{pmatrix}
\tau\\
\sigma
\end{pmatrix}
\right\rangle
&=\int_{\mathbb S^1}
g(\hat x)
\int_{\partial D}
\left(
\mathrm e^{-\mathrm i\kappa \hat x\cdot y}\tau(y)
+
\partial_{\nu_y}
\mathrm e^{-\mathrm i\kappa \hat x\cdot y}\,
\sigma(y)
\right)
ds(y)\,ds(\hat x)
\\
&=\int_{\partial D}
\Bigg[
\tau(y)
\int_{\mathbb S^1}
\mathrm e^{-\mathrm i\kappa \hat x\cdot y}
g(\hat x)\,ds(\hat x)
+
\sigma(y)
\int_{\mathbb S^1}
\partial_{\nu_y}
\mathrm e^{-\mathrm i\kappa \hat x\cdot y}
\,g(\hat x)\,ds(\hat x)
\Bigg]
ds(y).
\end{align*}
Hence,
\[
(\mathcal H^*)^\top g
=\begin{pmatrix}
\displaystyle
\int_{\mathbb S^1}
\mathrm e^{-\mathrm i\kappa \hat x\cdot y}
g(\hat x)\,ds(\hat x)
\\[2ex]
\displaystyle
\int_{\mathbb S^1}
\partial_{\nu_y}
\mathrm e^{-\mathrm i\kappa \hat x\cdot y}
\,g(\hat x)\,ds(\hat x)
\end{pmatrix}.
\]
We now define the reflection operator
\begin{align}\label{reflection_op}
\mathcal R : L^2(\mathbb S^1)\to L^2(\mathbb S^1)
\quad\text{which is given by}\quad 
(\mathcal R g)(\hat x)=g(-\hat x).
\end{align}
We see by this definition that $\mathcal R=\mathcal R^{-1}$.
Since
\[
\mathrm e^{-\mathrm i\kappa \hat x\cdot y}
=
\mathrm e^{\mathrm i\kappa (-\hat x)\cdot y},
\]
a change of variables $d=-\hat x$ shows that
\[
(\mathcal H^*)^\top
=
\mathcal H\mathcal R.
\]
We thus observe that the augmented near-field operator $4\kappa^4\mathcal Q\mathcal N\mathcal Q^\top\mathcal R$ obtains a symmetric factorization, which we summarize below.
\begin{Lemma}\label{Factorization_analysis}
Assume that $\kappa$ is not the fourth root of a Dirichlet eigenvalue of the bilaplacian $\Delta^2$ in $D$. Then the near-field operator defined by (\ref{NF_operator}) associated with the biharmonic scattering problem (\ref{eqnbcs})--(\ref{SRCs}) satisfies
\begin{align}\label{NF_fixed_factorization}
4\kappa^4\,\mathcal Q\mathcal N\mathcal Q^\top\mathcal R
&=
-\mathcal H^{*}\mathcal S^{-1}\mathcal H\quad \text{such that}\quad 4\kappa^4\mathcal Q\mathcal N\mathcal Q^\top\mathcal R: L^2(\mathbb S^1)\to L^2(\mathbb S^1),
\end{align}
where the bounded operators $\mathcal Q$, $\mathcal H$, $\mathcal S$, and $\mathcal R$ are given by (\ref{FF_transform}), (\ref{herglotz_op}), (\ref{coercive_op}), and (\ref{reflection_op}), respectively. 
\end{Lemma}
\begin{proof}
We recall that our assumption on the wavenumber $\kappa$ ensures that $\mathcal S^{-1}$ is a well-defined and bounded operator. Applying Corollary \ref{NF_factorization_thm}, we conclude that
\begin{align*}
4\kappa^4\mathcal Q\mathcal N\mathcal Q^\top \mathcal R&=4\kappa^4\mathcal Q(-\mathcal B\mathcal S^{-1}\mathcal B^\top)\mathcal Q^\top \mathcal R\\
&=-4\kappa^4(\mathcal Q\mathcal B)\mathcal S^{-1}(\mathcal B^\top \mathcal Q^\top)\mathcal R\\
&=-4\kappa^4\left(-\frac{1}{2\kappa^2}\mathcal H^{*}\right)\mathcal S^{-1}\left(-\frac{1}{2\kappa^2}\mathcal H^{*}\right)^\top\mathcal R\\
&=-\mathcal{H}^{*}\mathcal S^{-1}\mathcal H
\end{align*}
where we have used the fact that $\mathcal{R}=\mathcal{R}^{-1}$, proving the claim.
\end{proof}
We now proceed to relate our augmented near-field operator to the far-field operator corresponding to the biharmonic clamped scattering problem. To this end, for any $d \in \mathbb{S}^1$ define the plane wave incident field
\[
u^i(x,d)=\mathrm e^{\mathrm i\kappa x\cdot d} \,\,\text{where we denote the corresponding scattered field by } U^{s}(\cdot \, , d). 
\]
Here, the radiating scattered field $U^{s}(\cdot \, , d)$ satisfies \eqref{eqnbcs}-\eqref{SRCs} with the aforementioned plane wave incident field. The far-field operator associated with this biharmonic clamped scattering
problem is defined by
\begin{align}\label{FF_operator}
\mathcal F:L^2(\mathbb S^1)\to L^2(\mathbb S^1)
\quad \text{given by}\quad
(\mathcal Fg)(\hat x)
=
\int_{\mathbb S^1}
U^\infty(\hat x,d)g(d)\,ds(d),
\end{align}
where $U^\infty(\hat x,d)$ denotes the far-field pattern of the scattered
field $U^s(x,d)$.  By linearity of the direct problem \eqref{eqnbcs}-\eqref{SRCs}, the far-field data $\mathcal Fg$ can be interpreted as the far-field pattern corresponding to the boundary data $-\mathcal Hg$ with $\mathcal H$ defined by \eqref{herglotz_op}. It was shown in \cite{HarrisLiOzochiawaeze2026} that the far-field operator $\mathcal F$ is injective and has dense range if $\kappa$ is not a clamped transmission eigenvalue, i.e., if there exists no non-trivial solution to the following eigenvalue problem 
\begin{align}\label{eigvalue_problem}
\begin{dcases}
\Delta p-\kappa^2p=0 \quad \text{in } \mathbb R^2\setminus\overline D \quad \text{and}\quad \Delta q+\kappa^2q=0 &\quad \text{in } D, \\[2mm]
p+q=0 \quad \text{and}\quad \partial_\nu p+\partial_\nu q=0 &\quad \text{on } \partial D 
\end{dcases}
\end{align}
with $(p,q)\in H^1(\mathbb R^2\setminus\overline D)\times H^1(D)$. We introduce another auxiliary operator in relation to the far-field operator $\mathcal F$. We define the data-to-pattern operator
\begin{align}\label{G_operator}
\mathcal G : H^{3/2}(\partial D)\times H^{1/2}(\partial D) &\to L^2(\mathbb S^1) \quad \text{which is given by}\quad \mathcal G(\phi,\psi)^\top=v^{\infty},
\end{align}
where $v$ satisfies the exterior boundary value problem \eqref{ext_BVP}-\eqref{clamped_bcs_inhomog} given the boundary data $(\phi,\psi)^\top$ and $v^\infty$ denotes the
far-field pattern of $v$. By superposition we obtain the factorization 
\begin{align}
\mathcal F=-\mathcal G\mathcal H.
\end{align}
 We now introduce a new symmetric factorization of the far-field operator $\mathcal F$ for our purposes. We note this factorization is distinct from the factorization for $\mathcal F$ derived in [Theorem 3.5, \cite{Zhu2026}].
\begin{Theorem}
    Assume $\kappa$ is not the fourth root of the Dirichlet eigenvalue of the bilaplacian $\Delta^2$ in $D$. Then we obtain the following factorization for the far-field operator $\mathcal F$ corresponding to the biharmonic clamped scattering problem:
    \begin{align}\label{symm_fac2}
        \mathcal F&=\frac{1}{2\kappa^2}\mathcal H^{*}\mathcal S^{-1}\mathcal H.
    \end{align}
\end{Theorem}
\begin{proof}
The proof follows the same strategy as that of Theorem 3.5 in
\cite{Zhu2026}, which we provide here for completeness. We recall that the adjoint of the Herglotz wave operator $\mathcal H^{*}$ is given by 
    \[
    \mathcal H^{*}\begin{pmatrix}
        \tau\\
        \sigma
    \end{pmatrix}(x)=\int_{\partial D}\left(\mathrm e^{-\mathrm i\kappa \hat x\cdot y}   \tau(y)+\partial_{\nu_y} \mathrm e^{-\mathrm i\kappa \hat x\cdot y}\sigma(y) \right)\,ds(y),\quad \hat x\in \mathbb S^1.
    \]
    Define the function 
    \[
    v^s(x)=-2\kappa^2\int_{\partial D}\left(\mathbb G(x,y)\tau(y)+\partial_{\nu_y}\mathbb G(x,y)\sigma(y)\right)\,ds(y),\quad x\in \mathbb R^2\setminus\overline D.
    \]
    It then follows from the asymptotic behavior of the fundamental solution $\mathbb G(\cdot,y)$ that the far-field pattern $v^{\infty}$ of $v$ is $\mathcal H^{*}(\tau,\sigma)^\top$. By definition of the data-to-pattern $\mathcal G$ in \eqref{G_operator}, we obtain
    \[
    \mathcal H^{*}(\tau,\sigma)^\top =\mathcal G(v^s|_{\partial D},\partial_\nu v^s|_{\partial D})^\top =-2\kappa^2 \mathcal G\mathcal S(\tau,\sigma)^\top.
    \]
    Thus we obtain $\mathcal H^{*}=-2\kappa^2\mathcal G\mathcal S$. By our assumption on $\kappa$, we have $\mathcal S^{-1}$ is well-defined and bounded. Therefore, we obtain
    \[
   - \frac{1}{2\kappa^2}\mathcal H^{*}\mathcal{S}^{-1}=\mathcal G,
    \]
    which leads to the desired result by $\mathcal F=-\mathcal G\mathcal H$.
\end{proof}

With this result, we see that $4\kappa^4Q\mathcal N\mathcal Q^\top \mathcal R$ coincides with the far-field operator up to a constant scaling. This implies that this augmentation can be used to symmetrize the factorization of the near-field operator $\mathcal N$ and apply the results found in \cite{Zhu2026}. Note that this is similar to the idea found in \cite{Hu-nfFM}, where a so-called outgoing-to-incoming operator is used to symmetrize the factorization of the near-field operator associated with acoustic problems. Additionally, we remark that, in Theorem 3.11 of \cite{Zhu2026}, it was shown
that the orthonormal eigensystem of the far-field operator $\mathcal F=-2\kappa^2 Q\mathcal N\mathcal Q^\top \mathcal R$ can be used to uniquely recover the clamped obstacle $D$. This in turn says that the our (augmented) near-field operator can uniquely recover the obstacle. This is given in the following result.

\begin{Theorem}\label{THM_FM_NF_data}
    Assume that $\kappa$ is neither the fourth root of a Dirichlet eigenvalue of the bilaplacian $\Delta^2$ in $D$ nor a clamped transmission eigenvalue. Then we have that
    \begin{align}\label{range_test}
        z\in D\quad \text{if and only if}\quad \sum_{j=1}^\infty \frac{1}{|\lambda_j|}\left|\left(\mathrm e^{-\mathrm i\kappa \hat x\cdot z},\xi_j (\hat{x})\right)_{L^2(\mathbb S^1)}\right|^2<\infty,
    \end{align}
    where the pair $(\lambda_j, \xi_j)\in \mathbb C_{\neq 0}\times L^2(\mathbb S^1)$ consists of the eigenvalues and eigenfunctions of the augmented near-field operator $-2\kappa^2\mathcal Q\mathcal N\mathcal Q^\top \mathcal R$.
\end{Theorem}
\begin{proof}
By the factorizations \eqref{NF_fixed_factorization} and~\eqref{symm_fac2}, the augmented near-field operator satisfies
\begin{equation}\label{idenity}
-2\kappa^2\mathcal Q\mathcal N\mathcal Q^\top\mathcal R
=\mathcal F,
\end{equation}
where $\mathcal F$ denotes the far-field operator. It is known that, by excluding clamped transmission eigenvalues $\kappa$, we have that $\mathcal F$ is a compact, injective, and normal operator (see, for example, the results [Theorem 3.4, \cite{HarrisLeeLi2025}] and [Theorem 3.3(c), \cite{Zhu2026}]). The result is then a direct conscience of applying [Theorem 3.11, \cite{Zhu2026}]. 
\end{proof}
 This allows us to define a `computable' imaging functional to recover the obstacle from the measurements $\mathbb{M}$. Indeed, notice that from the range test property (\ref{range_test}), we obtain
\begin{align}\label{indicator_test}
    z\in D\quad \text{if and only if}\quad W(z)=\left[\sum_{j=1}^\infty \frac{1}{|\lambda_j|}\left|\left(\mathrm e^{-\mathrm i\kappa \hat x\cdot z},\xi_j (\hat{x})\right)_{L^2(\mathbb S^1)}\right|^2\right]^{-1}>0,
\end{align}
where $(\lambda_j, \xi_j)\in \mathbb C_{\neq 0}\times L^2(\mathbb S^1)$ is the orthonormal eigensystem of the augmented near-field operator $-2\kappa^2\mathcal Q\mathcal N\mathcal Q^\top \mathcal R$. This implies that 
\[
\chi_D(z)=\text{sign}(W(z))\coloneqq \begin{dcases}
    1,\quad z\in D\\
    0,\quad z\notin D.
\end{dcases}
\]
Since the near-field operator $\mathcal N$ is known and the operators $\mathcal Q$ and $\mathcal R$ can be precomputed, the indicator $W(z)$ defined in (\ref{indicator_test}) can be computed directly from the measured data. Therefore, from the near–field data we can
compute $W(z)$ and provide the contour plot to recover $D$.

Notice that since the near-field operator $\mathcal N$ is compact, this would imply that $-2\kappa^2\mathcal Q\mathcal N\mathcal Q^\top\mathcal R$ is compact as well. This would imply that the eigenvalues $|\lambda_j|\to 0$ rapidly as $j\to\infty$. By definition of the imaging function $W(z)$, we remark that dividing by these eigenvalues will cause numerical instabilities when trying to recover the clamped obstacle $D$. To ensure stability and a rigorous characterization of the clamped obstacle, we need to incorporate a regularization scheme to the factorization method as in, e.g., \cite{audibert2014lsm} and \cite{harris2023rfm}. More precisely, we assume that the measured near-field operator $\mathcal N^\delta$ has added random error such that
\[
||\mathcal N^\delta -\mathcal N||_{\mathcal L([L^2(\mathcal C)]^2)}=\mathcal O(\delta) \quad \text{implies}\quad ||\mathcal Q(\mathcal N^\delta-\mathcal N)\mathcal Q^\top \mathcal R||_{\mathcal L(L^2(\mathbb S^1))}=\mathcal O(\delta)\quad \text{as }\delta\to 0.
\]
This will be the case in many physical applications where the measured near-field data is polluted with random noise.

In order to have a stable reconstruction with noisy multi-static near-field data, the main result in \cite{harris2023rfm} implies that if we define the regularized imaging function
\begin{equation}
\label{reg_W}
W^{\mathrm{Reg}}(z;\alpha,\delta)
=
\left[
\sum_{j=1}^{\infty}
\frac{\phi^2(|\lambda_j^\delta|,\alpha)}
{|\lambda_j^\delta|}
\left|
\left(
\mathrm e^{-\mathrm ik\hat{x}\cdot z},
\xi_j^\delta(\hat{x})
\right)_{L^2(\mathbb{S}^1)}
\right|^2
\right]^{-1},
\end{equation}
 where $(\lambda_j^\delta, \xi_j^\delta)\in \mathbb C_{\neq 0}\times L^2(\mathbb S^1)$ are the eigenvalues and eigenvectors for the perturbed operator $-2 \kappa^2 \mathcal Q\mathcal N^\delta\mathcal Q^\top \mathcal R$, then we have that
\begin{align}\label{range_test_W_reg}
z\in D\quad \text{if and only if}\quad \liminf_{\alpha\to 0^{+}}\liminf_{\delta\to 0^{+}}W^{\mathrm{Reg}}(z,\alpha;\delta)>0.
\end{align}
Appealing to the continuity of the spectrum, we obtain that for any $z$,
\begin{align}
    \liminf_{\alpha\to 0^{+}}\liminf_{\delta\to 0^{+}}W^{\mathrm{Reg}}(z,\alpha;\delta)>0\quad \text{if and only if}\quad W(z)>0.
\end{align}
Here, it is assumed that $\phi(t,\alpha)$ is a non-negative continuous filter function derived from a regularization scheme such that
\begin{equation}\label{filter_fnc}
    \lim_{\alpha\to 0}\phi(t,\alpha)=1,\quad \phi(t,\alpha)\leq C_{\mathrm{Reg}}\quad \text{and}\quad \phi(t,\alpha)\leq C_{\alpha}t\quad \text{for all }t,\alpha>0,
\end{equation}
where the constant $C_{\mathrm{Reg}}$ is independent of the regularization parameter $\alpha$. For the numerical examples, we will employ two standard choices of
regularization filters, namely the Tikhonov filter and the spectral cutoff filter. They are defined by
\begin{equation}\label{tikn}
\phi_{\mathrm{Tik}}(t,\alpha)
=
\frac{t}{t^2+\alpha},
\end{equation}
and
\begin{equation}\label{cutoff}
\phi_{\mathrm{cut-off}}(t,\alpha)
=
\begin{cases}
1, & t>\alpha,\\
0, & t\leq \alpha,
\end{cases}
\end{equation}
respectively.

Thus, in the case of noisy measurements, it would be optimal to use the regularized imaging function $W^{\mathrm{Reg}}(z,\alpha;\delta)$ defined by \eqref{reg_W}. Though one can employ other common regularization schemes like Landweber and GLSM regularization, we will apply Tikhonov regularization and spectral cutoff due to its simplicity, stability, and well-understood spectral filtering properties.
\subsection{Precomputing the Far-Field Transform}
We observe that the operators $\mathcal Q$ and $\mathcal R$ are defined without the knowledge of the unknown obstacle $D$. This implies that they can be computed without any a priori knowledge of the scatterer. Indeed, for the reflection operator 
\[
(\mathcal Rg)(\hat x)=g(-\hat x)
\]
we note that
\[
g(\hat x)=g((\cos \theta,\sin \theta)),\quad \text{for}\quad\theta\in [0,2\pi),
\]
which can be viewed as a function of a single angular variable, i.e., $g(\hat x)$. Notice that the sum of angles formulas gives
\[
- \hat x=(\cos(\theta+\pi),\sin(\theta+\pi))\quad \text{for }\theta\in[0,2\pi).
\]
With this, we obtain, as in \cite{HarrisKleefeld2026}, that $(\mathcal Rg)(\theta)=g(\theta+\pi)$ and by some simple calculations
\[
(\mathcal Rg)(\theta)=\int_0^{2\pi}R_{\mathrm{func}}(\theta,\varphi)g(\varphi)\,d\varphi\quad \text{where}\quad R_{\mathrm{func}}(\theta,\varphi)=\frac{1}{2\pi}\sum_{|m|=0}^\infty \mathrm{e}^{\mathrm im(\theta-\varphi+\pi)}
\]
by using a Fourier series to express $g$ as has derived in \cite{nf-fft-dsm}.

For simplicity, we consider the case where $\mathcal C=\partial B_R$ is the boundary of a disk centered at the origin with radius $R>0$. We now precompute $\mathcal Q$ in a manner similar to \cite{HarrisNyugens2022}. We identify $\mathcal C$ with $(0,2\pi)$ via the parametrization $x=R(\cos\theta,\sin\theta)$ and write
\[
f(\theta),\, h(\theta)\in H^{3/2}(0,2\pi)\times H^{1/2}(0,2\pi).
\]
{We consider a radiating solution $w=w(r,\theta)\in H_{\text{loc}}^2(\mathbb R^2\setminus \overline B_R)$ of the exterior boundary value problem
\begin{align}\label{circle_bvp}
\begin{dcases}
\Delta^2w-\kappa^4w=0, & \text{in }\mathbb R^2\setminus\overline{B_R},\\[2mm]
w=f,\quad \partial_r w=h, & \text{on }\partial B_R.
\end{dcases}
\end{align}}
We expand the boundary data in Fourier series
\[
f(\theta)=\sum_{|m|=0}^\infty f_m \mathrm e^{\mathrm im\theta}
\quad \text{and} \quad
h(\theta)=\sum_{|m|=0}^{\infty} h_m \mathrm e^{\mathrm im\theta},
\]
with
\[
f_m=\frac{1}{2\pi}\int_0^{2\pi} f(\varphi)\mathrm e^{-\mathrm im\varphi}\,d\varphi
\quad\text{and}\quad
h_m=\frac{1}{2\pi}\int_0^{2\pi} h(\varphi)\mathrm e^{-\mathrm im\varphi}\,d\varphi.
\]
For each Fourier mode $m\in\mathbb Z$, we consider an ansatz for the unique radiating solution $w$ to \eqref{circle_bvp} is of the form
\[
w(r,\theta)=\sum_{|m|=0}^\infty w_m(r)\mathrm e^{\mathrm im\theta}.
\]
Using the factorization $\Delta^2-\kappa^4=(\Delta-\kappa^2)(\Delta+\kappa^2)$, the radial ODE admits two linearly independent solutions and gives that 
\[
w_m(r)=a_m H_m^{(1)}(\kappa r)+b_m H_m^{(1)}(\mathrm i\kappa r),
\]
where $H_m^{(1)}$ is the Hankel function of the first kind. To determine the coefficients we imposing the boundary conditions at $r=R$ gives, for each $m\in\mathbb Z$,
\[
w_m(R)=f_m,
\qquad
\partial_r w_m(R)=h_m,
\]
which yields the linear system
\[
\begin{pmatrix}
H_m^{(1)}(\kappa R) & H_m^{(1)}(\mathrm i\kappa R)\\[1ex]
\kappa (H_m^{(1)})'(\kappa R) &
\mathrm i\kappa {H_m^{(1)}}'(\mathrm i\kappa R)
\end{pmatrix}
\begin{pmatrix}
a_m\\
b_m
\end{pmatrix}
=
\begin{pmatrix}
f_m\\
h_m
\end{pmatrix}.
\]
Solving for $a_m$, we obtain
\[
a_m=
\frac{
\mathrm i\kappa {H_m^{(1)}}'(\mathrm i\kappa R) f_m
-
H_m^{(1)}(\mathrm i\kappa R) h_m
}{
\mathrm i\kappa H_m^{(1)}(\kappa R){H_m^{(1)}}'(\mathrm i\kappa R)
-
\kappa (H_m^{(1)})'(\kappa R)H_m^{(1)}(\mathrm i\kappa R)
}.
\]
Using the fact that 
\[
H_m^{(1)}(\kappa r)
=
\sqrt{\frac{2}{\pi \kappa r}}
\mathrm e^{\mathrm i(\kappa r - m\pi/2 - \pi/4)}
+O(r^{-3/2}),
\quad \text{and} \quad |H_m^{(1)}(\text{i}\kappa r)|
=O\left(\frac{ \text{e}^{- \kappa r } }{\sqrt{r}}\right),
 \quad \text{as} \, \, r\to\infty,
\]
together with the definition
\[
w(x)=\frac{\mathrm e^{\mathrm i\pi/4}}{\sqrt{8\pi \kappa}}
\frac{\mathrm e^{\mathrm i\kappa |x|}}{|x|^{1/2}}
\left\{
w^\infty(\hat x)+O(|x|^{-1})
\right\},
\]
we deduce that only the far-field pattern can be express as 
\[
w^\infty(\theta)
=
\frac{4}{\mathrm i}
\sum_{|m|=0}^\infty
a_m \mathrm e^{-\mathrm im\pi/2} \mathrm e^{\mathrm im\theta}.
\]
Substituting the expression for $a_m$, the far-field transform $\mathcal Q$ is given by
\[
(\mathcal Q (f,h)^\top)(\theta)
=
\frac{4}{\mathrm i}
\sum_{|m|=0}^\infty
\frac{
\mathrm i\kappa {H_m^{(1)}}'(\mathrm i\kappa R) f_m
-
H_m^{(1)}(\mathrm i\kappa R) h_m
}{
\mathrm i\kappa H_m^{(1)}(\kappa R){H_m^{(1)}}'(\mathrm i\kappa R)
-
\kappa (H_m^{(1)})'(\kappa R)H_m^{(1)}(\mathrm i\kappa R)
}
\mathrm e^{\mathrm im(\theta-\pi/2)}.
\]
We may therefore rewrite the far-field transform in kernel form as
\[
(\mathcal Q(f,h)^\top)(\theta)
=
\int_{0}^{2\pi}
Q_{\mathrm{func}}(\theta,\varphi)
\begin{pmatrix}
f(\varphi)\\
h(\varphi)
\end{pmatrix}
\,d\varphi.
\]
Here, the (vector-valued) kernel is defined by
\[
Q_{\mathrm{func}}(\theta,\varphi)
=
\left(\sum_{|m|=0}^\infty A_m \mathrm e^{\mathrm im(\theta-\varphi-\pi/2)},\sum_{|m|=0}^\infty B_m \mathrm e^{\mathrm im(\theta-\varphi-\pi/2)}\right),
\]
with Fourier coefficients
\begin{align}\label{first_coeff}
A_m &=
\frac{2}{\pi}
\frac{{H_m^{(1)}}'(\mathrm i\kappa R)}
{\mathrm i H_m^{(1)}(\kappa R){H_m^{(1)}}'(\mathrm i\kappa R)
-(H_m^{(1)})'(\kappa R)H_m^{(1)}(\mathrm i\kappa R)},
\end{align}
and
\begin{align}\label{second_coeff}
B_m &=
\frac{2\mathrm i}{\pi\kappa}
\frac{H_m^{(1)}(\mathrm i\kappa R)}
{\mathrm i H_m^{(1)}(\kappa R){H_m^{(1)}}'(\mathrm i\kappa R)
-(H_m^{(1)})'(\kappa R)H_m^{(1)}(\mathrm i\kappa R)}.
\end{align}
In the numerical evaluation of the imaging functional $W(z)$ given by (\ref{indicator_test}), we approximate $\mathcal Q$ by a truncated operator $\mathcal Q_M$ defined via Fourier truncation of its coefficient representation, i.e., given $M\in\mathbb N$, we define
\[
(Q_{M}(f,h)^\top)(\theta)=\int_0^{2\pi}Q_{\mathrm{func}}^M(\theta,\varphi)\begin{pmatrix}
    f(\varphi)\\
    h(\varphi)
\end{pmatrix}\,d\varphi
\]
where
\[
Q_{\mathrm{func}}^M(\theta,\varphi)
=
\left(\sum_{|m|=0}^M A_m \mathrm e^{\mathrm im(\theta-\varphi-\pi/2)},\sum_{|m|=0}^M B_m \mathrm e^{\mathrm im(\theta-\varphi-\pi/2)}\right),
\]
and $A_m$ and $B_m$ refer to \eqref{first_coeff} and $\eqref{second_coeff}$, respectively.
\begin{Theorem}\label{convergence_THM}
Let
$\mathcal Q:H^{3/2}(0,2\pi)\times H^{1/2}(0,2\pi)
\rightarrow L^2(\mathbb S^1)$ be the far-field transform defined in \eqref{FF_transform}, and let
$\mathcal Q_M$ be the Fourier truncation obtained by retaining only the
modes $|m|\leq M$. Then
\[
\|\mathcal Q-\mathcal Q_M\|_{\mathcal L(H^{3/2}\times H^{1/2},L^2)}
=
\mathcal O\left(
\frac{1}{2^M}
\right),
\qquad M\rightarrow\infty .
\]
In particular,
$
\mathcal Q_M\rightarrow \mathcal Q$ in operator norm.
\end{Theorem}
\begin{proof}
Let $(f,h)^{\top}\in H^{3/2}(0,2\pi)\times H^{1/2}(0,2\pi)$
with Fourier expansions
\[
f(\theta)=\sum_{|m|=0}^\infty f_m\mathrm e^{\mathrm im\theta},
\qquad
h(\theta)=\sum_{|m|=0}^\infty h_m\mathrm e^{\mathrm im\theta}.
\]
By the Fourier representation of $\mathcal Q$, we obtain
\[
(\mathcal Q-\mathcal Q_M)(f,h)^\top
=
\sum_{|m|=M+1}^\infty
(A_mf_m+B_mh_m)\mathrm e^{\mathrm im(\theta-\pi/2)}.
\]
Hence, by Parseval's identity,
\[
\|(\mathcal Q-\mathcal Q_M)(f,h)^\top\|_{L^2}^2
\leq
C\sum_{|m|=M+1}^\infty
\left(
|A_m|^2|f_m|^2+
|B_m|^2|h_m|^2
\right).
\]
In order to estimate further, we now use the facts that
\[
(H_m^{(1)})'(t) = \frac{1}{2}\left(H_{m-1}^{(1)}(t)-H_{m+1}^{(1)}(t)\right)
\]
for all $m\in \mathbb Z$ as well as the asymptotic results (see for e.g. \cite{NIST:DLMF})
\[
-\mathrm iH_m^{(1)}(t) \sim \sqrt{\frac{2}{\pi m}}\left(\frac{\mathrm et}{2m}\right)^{-m}
\quad \text{and} \quad
-\mathrm i(H_m^{(1)})'(t) \sim \frac{m}{t}\sqrt{\frac{2}{\pi m}}\left(\frac{\mathrm et}{2m}\right)^{-m}\quad \text{as }m\to \infty
\]
for any fixed $t$. Applying these estimates for $t=\kappa R$ and $t=\mathrm i\kappa R$ in the expressions for $A_m$ and $B_m$ yields the result in Lemma \ref{lem:coeff_asymptotics}, i.e., that
\[
|A_m|^2+|B_m|^2
\leq
C\frac{\pi m}{2\kappa^2}
\left(\frac{\mathrm e\kappa R}{2m}\right)^{2m}.
\]
Consequently,
\[
\|(\mathcal Q-\mathcal Q_M)(f,h)^\top\|_{L^2}^2
\leq
C\sum_{|m|\geq M+1}
\frac{\pi m}{2\kappa^2}
\left(\frac{\mathrm e\kappa R}{2m}\right)^{2m}
\left(|f_m|^2+|h_m|^2\right).
\]
Writing
\[
\frac{\pi m}{2\kappa^2}
\left(\frac{\mathrm e\kappa R}{2m}\right)^{2m}
=
\frac{1}{4^m}
\frac{\pi m}{2\kappa^2}
\left(\frac{\mathrm e\kappa R}{m}\right)^{2m},
\]
we consider the sequence
\[
b_m:=
\frac{\pi m}{2\kappa^2}
\left(\frac{\mathrm e\kappa R}{m}\right)^{2m}.
\]
Its consecutive-term ratio satisfies
\[
\frac{b_{m+1}}{b_m}
=
\frac{m+1}{m}
\left(\frac{m}{m+1}\right)^{2m}
\left(\frac{\mathrm e\kappa R}{m+1}\right)^2
\longrightarrow 0
\qquad\text{as }m\to\infty.
\]
Hence, by the ratio test, the sequence $(b_m)_{m\geq 1}$ is bounded.
Thus, for $M$ sufficiently large and $|m|\geq M+1$,
\[
\frac{\pi m}{2\kappa^2}
\left(\frac{\mathrm e\kappa R}{2m}\right)^{2m}
\leq
\frac{C}{4^m}
\leq
\frac{C}{4^{M+1}}.
\]
It follows that
\[
\begin{aligned}
\|(\mathcal Q-\mathcal Q_M)(f,h)^\top\|_{L^2}^2
&\leq
\frac{C}{4^{M+1}}
\sum_{|m|\geq M+1}
\left(|f_m|^2+|h_m|^2\right)\\
&\leq
\frac{C}{4^{M+1}}
\|(f,h)\|_{H^{3/2}\times H^{1/2}}^2
\end{aligned}
\]
by the Sobolev embeddings
$H^{3/2}(0,2\pi)\hookrightarrow L^2(0,2\pi)$ and
$H^{1/2}(0,2\pi)\hookrightarrow L^2(0,2\pi)$. Therefore, taking the square root of both sides above, we can conclude that
\[
\|(\mathcal Q-\mathcal Q_M)(f,h)^\top\|_{L^2}
\leq
\frac{C}{2^{M+1}}
\|(f,h)\|_{H^{3/2}\times H^{1/2}}^2,
\] 
which proves the claim.
\end{proof}
In practice, the convergence rate in Theorem \ref{convergence_THM} means that we can approximate the operator $\mathcal{Q}$ by the truncated series $\mathcal{Q}_M$. Therefore, due to the fast convergence one does not need to keep many terms in the series. This will allow one to evaluate the approximate operator $\mathcal{Q}_M$ more efficiently when numerically solving the inverse problem.
\section{Numerical Experiments}\label{numerics_sect}
\subsection{The Method of Fundamental Solutions for the Scattering Problem}
In this section, we will provide numerical reconstructions of the clamped obstacle $D$ using the regularized imaging function $W^{\mathrm{Reg}}(z,\alpha;\delta)$ given by \eqref{reg_W}. We will see that the factorization method provides quality reconstructions of the clamped obstacle with augmented multistatic near-field data, i.e., the post-processing via the far-field transform leads to a stable indicator for recovering the scatterer. To this end, we first generate synthetic near-field data by numerically solving the direct scattering problem \eqref{eqnbcs}-\eqref{SRCs} using the method of fundamental solutions developed in \cite{karageorghislesnic2024}. Compared with classical boundary integral equation (BIE) methods for scattering problems, the method of fundamental solutions (MFS) offers a meshless and computationally efficient alternative that is particularly well-suited for repeated forward evaluations in inverse problem settings. In contrast to BIE formulations, which require boundary discretization and the numerical treatment of singular kernels, MFS avoids surface meshing by representing the solution using fundamental solutions with sources placed outside the physical domain. This leads to a simpler implementation and reduced geometric preprocessing, while still achieving accurate approximations for smooth boundary geometries.  

We take the fourth MFS approach specified in \cite{karageorghislesnic2024}. That is, we approximate the solution of \eqref{eqnbcs}-\eqref{SRCs} by a linear combination of non-singular fundamental solutions
\begin{equation}\label{MFS_4}
    u_{N}^s(x)=\sum_{j=1}^N c_j\Phi_{\mathrm H}(x,y_j)+\sum_{j=1}^Nd_j\Phi_{\mathrm M}(x,y_j),\quad x\in (\mathbb R^2\setminus\overline D)\cup \partial D,
\end{equation}
where $(y_j)_{j=\overline{1,N}}$ are the source points located inside $D$ and $(c_j)_{j=\overline{1,2N}}$ are the unknown complex coefficients to be determined by imposing the boundary conditions \eqref{eqnbcs}. In \eqref{MFS_4}, the pair $(\Phi_{\mathrm H},\Phi_{\mathrm M})$ are the fundamental solutions for the propagative and evanescent components of the fundamental solution $\mathbb G(\cdot,y)$, respectively. That is, 
\[
\Phi_{\mathrm H}(x,y)=\frac{\mathrm i}{4}H_0^{(1)}(\kappa|x-y|)\quad \text{and}\quad \Phi_{\mathrm M}(x,y)=\frac{\mathrm i}{4}H_{0}^{(1)}(\mathrm i\kappa|x-y|)\quad\text{for } x\neq y
\]
satisfy
\[
\Delta \Phi_{\mathrm H}(x,y)+\kappa^2\Phi_{\mathrm H}(x,y)=-\delta(x-y)\quad \text{and}\quad \Delta \Phi_{\mathrm M}(x,y)-\kappa^2\Phi_{\mathrm M}(x,y)=-\delta(x-y)\quad \text{in }\mathbb R^2.
\]
Here, $H_0^{(1)}$ is the Hankel function of the first kind of order zero.
Assume that $D \subset \mathbb{R}^2$ is a smooth, star-shaped domain with respect to the origin. In polar coordinates, its boundary $\partial D$ can be parameterized as
\begin{equation}
\label{eq:param_boundary}
x_1 = r(\theta)\cos\theta, 
\qquad 
x_2 = r(\theta)\sin\theta,
\qquad \theta \in [0,2\pi),
\end{equation}
where $r$ is a smooth $2\pi$-periodic function. We discretize the boundary $\partial D$ using $M$ equispaced collocation points defined by
\begin{equation}
\label{eq:collocation}
x_m = r(\tilde{\theta}_m)\big(\cos\tilde{\theta}_m, \sin\tilde{\theta}_m\big),
\qquad
\tilde{\theta}_m = \frac{2\pi(m-1)}{M},
\quad m=1,\dots,M.
\end{equation}
We also place $N$ source points on a ``pseudo-boundary'' $\partial D^\prime$, defined by a radial contraction of $\partial D$:
\begin{equation}
\label{eq:mfs1}
y_\ell = \eta_1\, r(\theta_\ell)\big(\cos\theta_\ell, \sin\theta_\ell\big),
\qquad
\theta_\ell = \frac{2\pi(\ell-1)}{N}, \quad \ell=1,\dots,N,
\end{equation}
and another $N$ sources on a pseudo-boundary $\partial D^{\prime\prime}$ given by 
\begin{equation}
    \label{eq:mfs2}
    y_{N+\ell}=\eta_2 r(\theta_\ell)(\cos{\theta_\ell},\sin{\theta_\ell}),\quad \ell=1,\dots,N,
\end{equation}
where $\eta_1,\eta_2 \in (0,1)$ are contraction parameters with $\eta_1\neq \eta_2$ and $\theta_\ell=2\pi (\ell-1)/N$. Additionally, for our MFS system we make the artificial assumption that $\kappa^2$ is not an interior Dirichlet eigenvalue of $-\Delta$ in $D^\prime$.
The imposition of the clamped boundary conditions yields a $2M\times 2N$ linear system of the form
\begin{equation}
\left(
\begin{array}{c|c}
A_{11} & A_{12} \\
\hline
A_{21} & A_{22}
\end{array}
\right)\begin{pmatrix}
\mathbf{c} \\
\hline
\mathbf{d}
\end{pmatrix}=\begin{pmatrix}
b_1 \\
\hline
b_2
\end{pmatrix}.
\end{equation}
The matrices $A_{11}, A_{12}, A_{21}, A_{22} \in \mathbb{R}^{M \times N}$ are defined by
\begin{align*}
A_{11} =  \, \Phi_{\mathrm{H}}(x_i, y_j), 
\quad\text{and}\quad
A_{12} =  \, \Phi_{\mathrm{M}}(x_i, y_j),
\end{align*}
\begin{align*}
A_{21} = \partial_{\nu_{x}} \, \Phi_{\mathrm{H}}(x_i, y_j), 
\quad \text{and}\quad
A_{22} = \partial_{\nu_{x}} \, \Phi_{\mathrm{M}}(x_i, y_j),
\end{align*}
for $i=1,\cdots,M$ and $j=1,\cdots, N$. Moreover, the vectors $b_1,b_2\in \mathbb R^{M}$ are defined by
\[
b_1 = -\mathbb{G}(x_i, y), 
\qquad
b_2 = -\partial_{\nu_{x}} \mathbb{G}(x_i, y),
\qquad i = 1,\dots,M.
\]
Having determined the unknown coefficient vectors $\mathbf c=(c_j)_{j=\overline{1,N}}\in \mathbb C^{N}$ and $\mathbf d=(d_j)_{j=\overline{1,N}}\in \mathbb C^{N}$, the approximation \eqref{MFS_4} may be calculated anywhere in $(\mathbb R^2\setminus\overline D)\cup \partial D$.
\subsection{Numerical Implementation}
Now we are ready to provide some numerical examples that have been implemented in \textsf{MATLAB} 2024a.
We first discuss how the synthetic near-field matrix is computed. We will take
\[
\mathcal C=\partial B_R\quad \text{with }R=2 \text{ in our numerical experiments }
\]
for all our examples. 

The sources and receivers are located at the points 
\[
X_k=Y_k=R(\cos{\varphi_k},\sin{\varphi_k}),
\qquad k=1,\dots,40,
\]
where the angles $\varphi_k\in[0,2\pi)$ are equally spaced. The measured near-field data consist of the scattered field and its normal derivative,
\[
u_N^s(x,y),\quad \partial_{\nu_x}u_N^s(x,y),\quad
\widetilde{u}_N^s(x,y),\quad
\partial_{\nu_x}\widetilde{u}_N^s(x,y),
\qquad (x,y)\in\mathcal{C}\times\mathcal{C}.
\]
The multistatic near-field matrix $\mathbf{N}\in\mathbb{C}^{80\times80}$ is then given by
\begin{equation}\label{NFmatrix}
\mathbf{N}
=
\begin{pmatrix}
\big(u_N^s(X_k,Y_\ell)\big)_{k,\ell=1}^{40}
&
\big(\widetilde{u}_N^s(X_k,Y_\ell)\big)_{k,\ell=1}^{40}
\\[2mm]
\big(\partial_{\nu_x}u_N^s(X_k,Y_\ell)\big)_{k,\ell=1}^{40}
&
\big(\partial_{\nu_x}\widetilde{u}_N^s(X_k,Y_\ell)\big)_{k,\ell=1}^{40}
\end{pmatrix}.
\end{equation}
We can additionally introduce noisy measurements by defining
\begin{equation}
\mathbf{N}^{\delta}
=
\begin{pmatrix}\label{noisy_NFmatrix}
\big((u_N^s)^\delta(X_k,Y_\ell)\big)_{k,\ell=1}^{40}
&
\big((\widetilde{u}_N^s)^\delta(X_k,Y_\ell)\big)_{k,\ell=1}^{40}
\\[2mm]
\big((\partial_{\nu_x}u_N^s)^\delta(X_k,Y_\ell)\big)_{k,\ell=1}^{40}
&
\big((\partial_{\nu_x}\widetilde{u}_N^s)^\delta(X_k,Y_\ell)\big)_{k,\ell=1}^{40}
\end{pmatrix},
\end{equation}
where $0<\delta<1$ denotes the noise level. The perturbed measurements are defined by
\begin{align}
(u_N^s)^\delta(X_k,Y_\ell)
&=(1+\delta E^{(1)}_{k\ell})u_N^s(X_k,Y_\ell),
&
(\widetilde{u}_N^s)^\delta(X_k,Y_\ell)
&=(1+\delta \widetilde E^{(1)}_{k\ell})
\widetilde{u}_N^s(X_k,Y_\ell),
\nonumber\\
(\partial_{\nu_x}u_N^s)^\delta(X_k,Y_\ell)
&=(1+\delta E^{(2)}_{k\ell})
\partial_{\nu_x}u_N^s(X_k,Y_\ell),
&
(\partial_{\nu_x}\widetilde{u}_N^s)^\delta(X_k,Y_\ell)
&=(1+\delta \widetilde E^{(2)}_{k\ell})
\partial_{\nu_x}\widetilde{u}_N^s(X_k,Y_\ell),
\label{noise_data}
\end{align}
for $k,\ell=1,\dots,40$. Here,
\[
E^{(1)},E^{(2)},\widetilde E^{(1)},\widetilde E^{(2)}
\in\mathbb R^{40\times40}
\]
are random noise matrices satisfying
\[
\|E^{(i)}\|_2=\|\widetilde E^{(i)}\|_2=1,
\qquad i=1,2.
\]
In order to apply the reconstruction method in \eqref{range_test_W_reg}, we need to discretize the operators $\mathcal Q$ and $\mathcal R$. 
Note that we define the truncated series approximations of the kernel functions by
\[
Q_{\mathrm{func}}(\theta,\varphi)
=
\left(
\sum_{|m|=0}^{10}A_m \mathrm e^{\mathrm im(\theta-\varphi-\pi/2)},
\;
\sum_{|m|=0}^{10}B_m \mathrm e^{\mathrm im(\theta-\varphi-\pi/2)}
\right)
\quad\text{and}\quad
R_{\mathrm{func}}(\theta,\varphi)
=
\frac{1}{2\pi}
\sum_{|m|=0}^{10}
\mathrm e^{\mathrm im(\theta-\varphi+\pi)},
\]
where $A_m$ and $B_m$ are given by \eqref{first_coeff} and \eqref{second_coeff}, respectively. The corresponding discrete matrices
$\mathbf Q,\mathbf R\in\mathbb{C}^{80\times80}$
are obtained by evaluating these kernel functions at the measurement directions
associated with the compound data space. More precisely,
\[
\mathbf Q=
\begin{bmatrix}
\bigl[Q_{\mathrm{func}}^{(1)}(\varphi_k,\varphi_\ell)\bigr]_{k,\ell=1}^{40} & 0\\[2mm]
0 &\bigl[Q_{\mathrm{func}}^{(2)}(\varphi_k,\varphi_\ell)\bigr]_{k,\ell=1}^{40}
\end{bmatrix}
\quad\text{and}\quad
\mathbf R=
\bigl[R_{\mathrm{func}}(\theta_k,\theta_\ell)\bigr]_{k,\ell=1}^{80}.
\]
 With this, we can compute the imaging function $W^{\mathrm{Reg}}(z)$ given in \eqref{reg_W} by first computing
\[
4\kappa^4\mathbf Q\mathbf N\mathbf Q^\top \mathbf R\quad \text{and}\quad \mathbf b_z=[\mathrm e^{-\mathrm i\kappa z\cdot \hat x_1},\cdots, \mathrm e^{-\mathrm i\kappa z\cdot \hat x_{80}}]^\top
\]
with $\hat x_k=(\cos \theta_k,\sin \theta_k)$ with $k=1,\cdots, 80$.
We will provide numerical examples of the reconstructions, where $W^{\mathrm{Reg}}(z)$ is computed via
\begin{equation}
\label{reg2_W}
W^{\mathrm{Reg}}(z;\alpha,\delta)
=
\left[
\sum_{j=1}^{80}
\frac{
\phi^2(|\lambda_j^\delta|,\alpha)
}{
|\lambda_j^\delta|
}
\left|
\mathbf{u}_j\cdot \mathbf{b}_z
\right|^2
\right]^{-1}
\end{equation}
using the built-in $\mathbf{eig}$ solver in \textsf{MATLAB} to compute $\lambda_j$ and $\mathbf u_j$, i.e., the eigenvalues and eigenvectors of $\mathbf Q\mathbf N\mathbf Q^\top\mathbf R$. For our examples, we will use the Tikhonov and spectral-cutoff filter functions given by \eqref{tikn} and \eqref{cutoff}, respectively.

We can now recover example scatterers with the synthetic near-field data. We use three obstacle shapes in our experiments: a disk with a fixed radius $=0.75$, a five-pointed star-shaped obstacle, and a peanut-shaped obstacle. The boundary parametrizations of all three clamped obstacles are shown in Table \ref{tab:test_obstacles}. In Figure \ref{fig:domains} we provide a visualization of the scatterers defined in Table \ref{tab:test_obstacles}.  In these examples, the sampling region is $[-2,2]^2$ and we select $120\times 120$ equally spaced grid points within this region. We present contour plots of the imaging function for both scatterers at different noise levels to assess the stability of our method. In all figures, the white dotted line represents the boundary $\partial D$ of the obstacle. The imaging function is normalized in the sampling region to have a maximum value of one.
\begin{table}[h]
\centering
\caption{The boundary parametrizations of $\partial D=\gamma(t)$ for
$t\in[0,2\pi)$.}
\label{tab:test_obstacles}
\begin{tabular}{ll}
\hline
Obstacle & Boundary parameterization $\gamma(t)$ \\
\hline
Disk &
$\gamma(t)=0.75(\cos t,\sin t)^\top$ \\[0.8ex]

Five-pointed star &
$\gamma(t)=\bigl(1+0.3\cos(5t)\bigr)(\cos t,\sin t)^\top$ \\[0.8ex]

Peanut &
$\gamma(t)=1.2\sqrt{\sin^2 t+0.1\cos^2 t}\,(\cos t,\sin t)^\top$ \\
\hline
\end{tabular}
\end{table}

\begin{figure}[ht]
\centering
\begin{subfigure}[t]{0.22\textwidth}
    \centering
    \includegraphics[width=\linewidth]{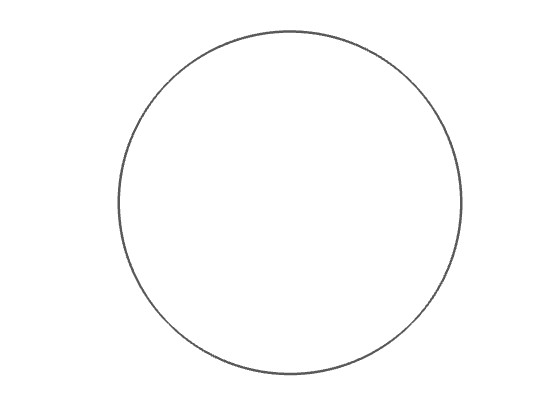}
    \caption{Disk}
\end{subfigure}
\begin{subfigure}[t]{0.22\textwidth}
    \centering
    \includegraphics[width=\linewidth]{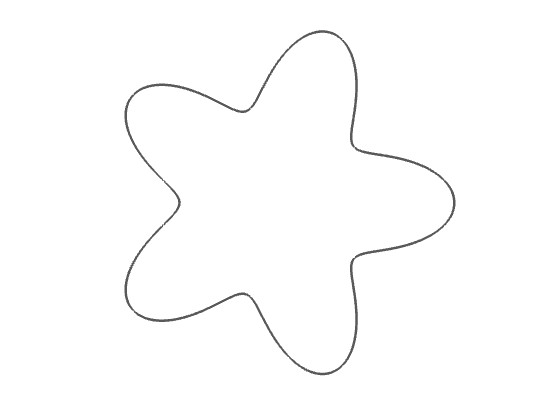}
    \caption{Five-pointed star}
\end{subfigure}
\begin{subfigure}[t]{0.22\textwidth}
    \centering
    \includegraphics[width=\linewidth]{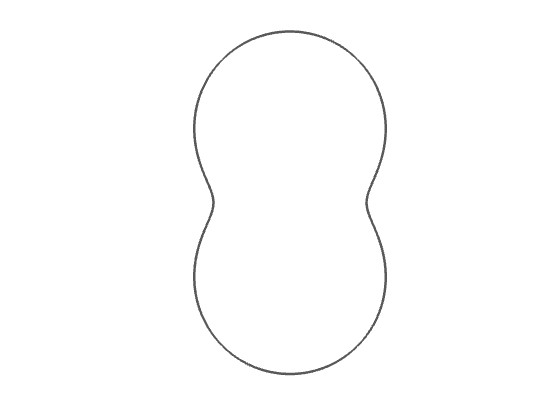}
    \caption{Peanut}
\end{subfigure}
\caption{Reference scatterer geometries defined in Table \ref{tab:test_obstacles}.}
\label{fig:domains}
\end{figure}
\subsection{Example 1: Resolution}
The main task of this subsection is to illustrate the performance of the imaging function at different frequencies. In particular, we choose a trajectory of the wavenumbers $\kappa=2, 4,$ and $8$ in our examples. For the factorization method, we employ Tikhonov regularization with the fixed parameter $\alpha=10^{-6}$, chosen empirically based on its effectiveness in recovering the original shape. A more systematic choice of the regularization parameter is discussed later.

We now present numerical examples of the disk-shaped, five-pointed star-shaped and peanut-shaped obstacles using our imaging function. Figure~\ref{fig:disk_resolution} illustrates reconstructions of the disk-shaped
scatterer with increasing wavenumber, using a fixed radius $=0.75$ and noise-free data. The reconstructions are similar for
$\kappa=2$ and $\kappa=4$, while the improvement in resolution is more
noticeable for $\kappa=8$. 

Having demonstrated the effectiveness of the proposed method for the disk-shaped scatterer, we next consider more geometrically complex scatterers and investigate its numerical performance under resolution. Figure~\ref{fig:star_resolution} illustrates the effect of increasing the wavenumber on the reconstruction of the five-pointed star. In all cases, the general size and overall shape of the star-shaped obstacle is recovered. At higher frequency, we observe more detailed reconstructions of the obstacle. At the same time, the indicator exhibits lower values over a larger portion of the imaging domain at higher frequencies, producing darker regions in the visualization. Similarly, Figure ~\ref{fig:peanut_resolution} demonstrates the effect of increasing the wavenumber of the peanut-shaped obstacle. Though the change in wavenumber slightly improves the resolution of the reconstructions, there is a slightly noticeable sensitivity of the reconstruction when the wavenumber is increased for both star-shaped and peanut-shaped reconstructions.

\begin{figure}[ht]
\centering

\begin{subfigure}[t]{0.33\textwidth}
    \centering
    \includegraphics[width=\linewidth]{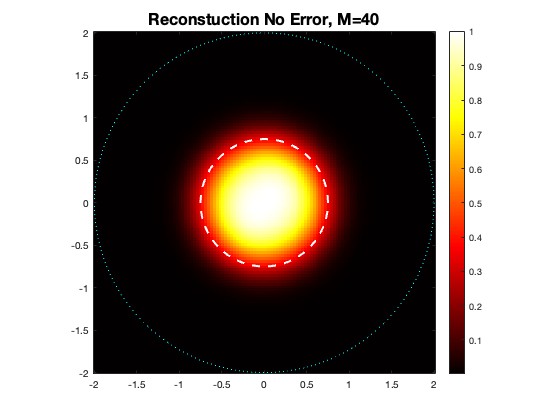}
    \caption{$\kappa = 2$}
\end{subfigure}\hfill
\begin{subfigure}[t]{0.33\textwidth}
    \centering
    \includegraphics[width=\linewidth]{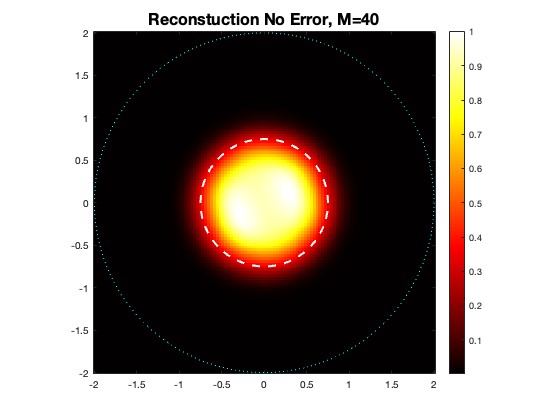}
    \caption{$\kappa = 4$}
\end{subfigure}\hfill
\begin{subfigure}[t]{0.33\textwidth}
    \centering
    \includegraphics[width=\linewidth]{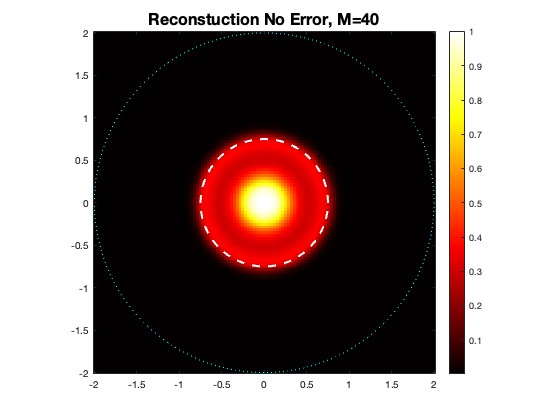}
    \caption{$\kappa = 8$}
\end{subfigure}

\caption{Reconstruction of the disk-shaped scatterer with a fixed radius $=0.75$ by the imaging function $W^{\mathrm{Reg}}(z)$ for increasing wavenumbers with no noise. Here $M=40$ refers to the number of sources and receivers.}
\label{fig:disk_resolution}
\end{figure}

\begin{figure}[ht]
\centering

\begin{subfigure}[t]{0.33\textwidth}
    \centering
    \includegraphics[width=\linewidth]{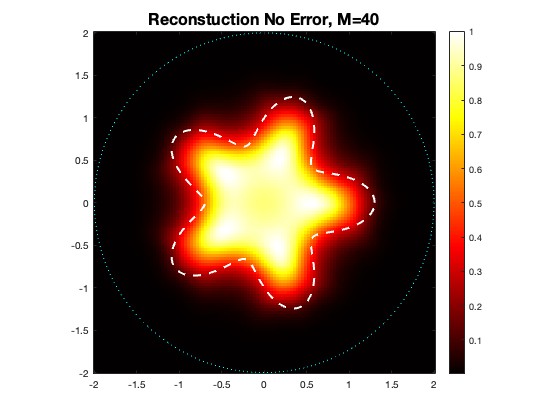}
    \caption{$\kappa = 2$}
\end{subfigure}\hfill
\begin{subfigure}[t]{0.33\textwidth}
    \centering
    \includegraphics[width=\linewidth]{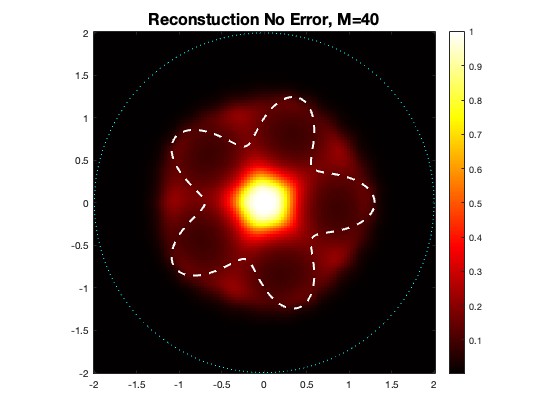}
    \caption{$\kappa = 4$}
\end{subfigure}\hfill
\begin{subfigure}[t]{0.33\textwidth}
    \centering
    \includegraphics[width=\linewidth]{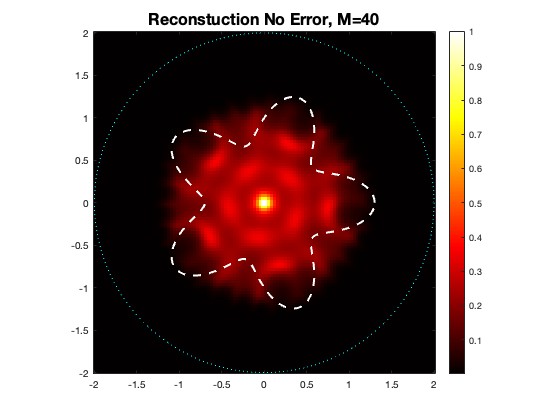}
    \caption{$\kappa = 8$}
\end{subfigure}

\caption{Reconstruction of the five-pointed star-shaped scatterer by the imaging function $W^{\mathrm{Reg} }(z)$ for increasing wavenumbers with no noise. Here $M=40$ refers to the number of sources and receivers.}
\label{fig:star_resolution}
\end{figure}

\begin{figure}[ht]
\centering

\begin{subfigure}[t]{0.33\textwidth}
    \centering
    \includegraphics[width=\linewidth]{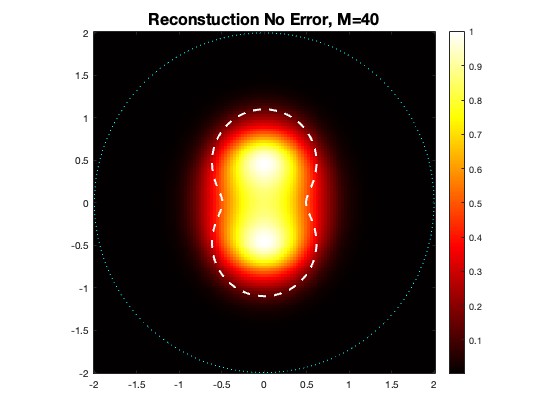}
    \caption{$\kappa = 2$}
\end{subfigure}\hfill
\begin{subfigure}[t]{0.33\textwidth}
    \centering
    \includegraphics[width=\linewidth]{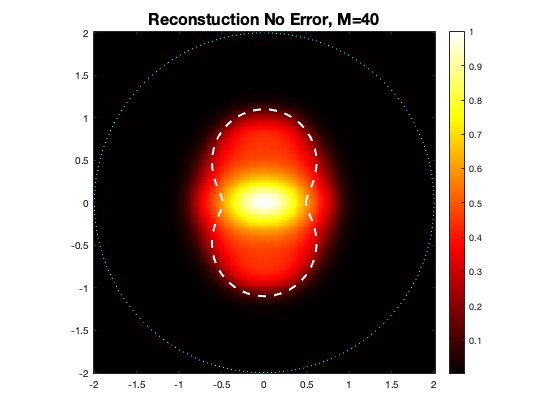}
    \caption{$\kappa = 4$}
\end{subfigure}\hfill
\begin{subfigure}[t]{0.33\textwidth}
    \centering
    \includegraphics[width=\linewidth]{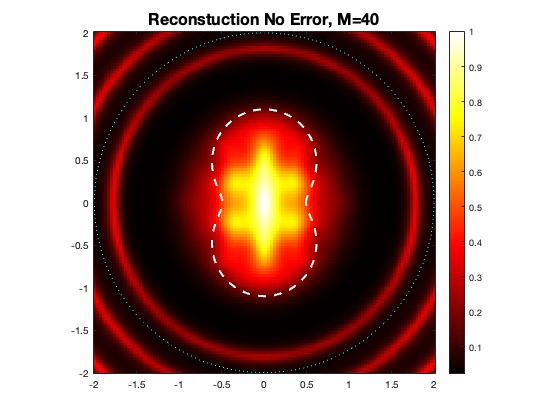}
    \caption{$\kappa = 8$}
\end{subfigure}

\caption{Reconstruction of the peanut-shaped scatterer by the imaging function $W^{\mathrm{Reg}}(z)$ for increasing wavenumbers with no noise. Here $M=40$ refers to the number of sources and receivers.}
\label{fig:peanut_resolution}
\end{figure}

\subsection{Example 2: Influence of noise in the data}
To examine the stability of the proposed method, we consider increasing levels of noise in the near-field data at a fixed heuristic choice for a  regularization parameter of $\alpha=10^{-6}$ with Tikhonov regularization. Figure~\ref{fig:disk_noise} shows reconstructions of the disk-shaped obstacle for increasing levels of noise in the near-field matrix, using the fixed regularization parameter $\alpha=10^{-6}$ at the fixed wavenumber $\kappa=2$. The reconstructions remain relatively stable as the noise level increases, with only modest degradation visible at the highest noise level. We see similar results for the five-pointed star and peanut shaped obstacles.  In Figure~\ref{fig:star5_noise}, we recover the five-pointed star-shaped obstacle at the fixed wavenumber $\kappa=2$. The reconstructions remain largely unaffected under $2\%$ and $5\%$ noise, while a  degradation in the imaging quality becomes noticeable at the $20\%$ noise level. In particular, the boundary profile becomes less sharp and the shape reconstruction degrades slightly compared to the exact boundary. Nevertheless, the overall shape and location of the five-pointed star-shaped scatterer are well preserved, demonstrating the robustness of the method with respect to moderate measurement perturbations. Figure~\ref{fig:peanut_noise} shows a similar trajectory for peanut-shaped obstacle reconstructions at increasing noise levels. The reconstructions remain robust under noise overall, with noticeable degradation for the case of $20\%$ noise.

\begin{figure}[h]
\centering
\begin{subfigure}[t]{0.33\textwidth}
    \centering
    \includegraphics[width=\linewidth]{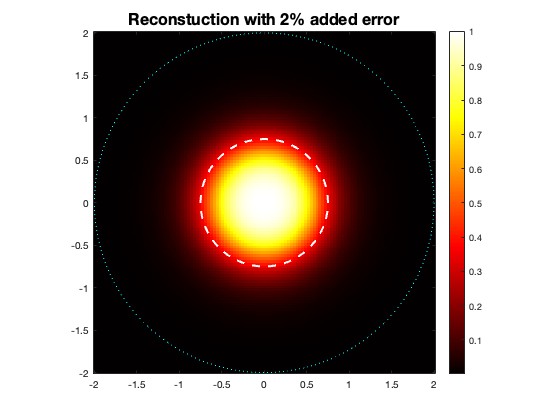}
    \caption{$2\%$ noise}
\end{subfigure}\hfill
\begin{subfigure}[t]{0.33\textwidth}
    \centering
    \includegraphics[width=\linewidth]{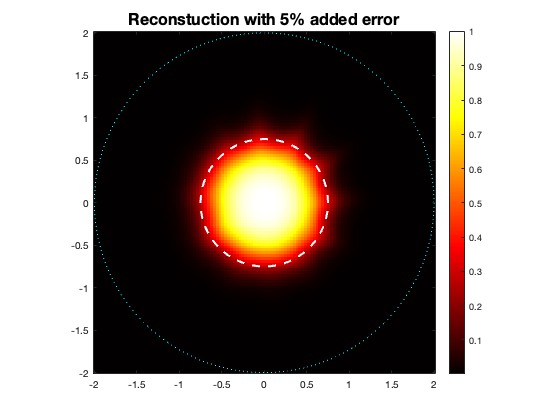}
    \caption{$5\%$ noise}
\end{subfigure}\hfill
\begin{subfigure}[t]{0.33\textwidth}
    \centering
    \includegraphics[width=\linewidth]{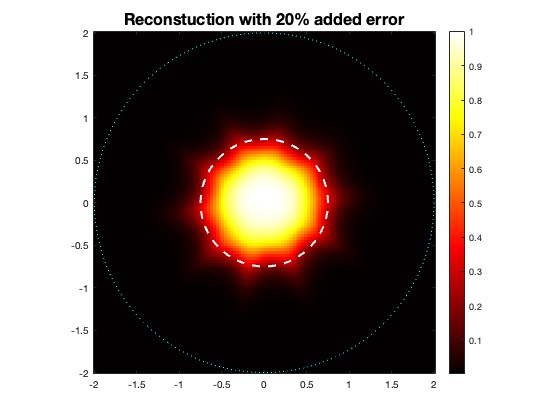}
    \caption{$20\%$ noise}
\end{subfigure}
\caption{Reconstruction of the disk-shaped scatterer at a fixed wavenumber $\kappa=2$ and fixed radius $=0.75$ using the imaging function
$W^{\mathrm{Reg}}(z)$ under increasing levels of noise in the near-field data. We used $40$ sources and receivers.}
\label{fig:disk_noise}
\end{figure}
\begin{figure}[h]
\centering

\begin{subfigure}[t]{0.33\textwidth}
    \centering
    \includegraphics[width=\linewidth]{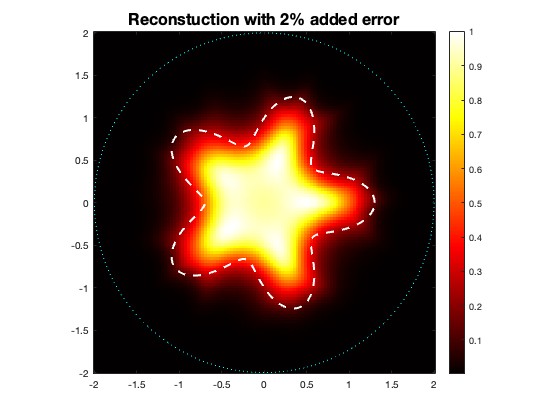}
    \caption{$2\%$ noise}
\end{subfigure}\hfill
\begin{subfigure}[t]{0.33\textwidth}
    \centering
    \includegraphics[width=\linewidth]{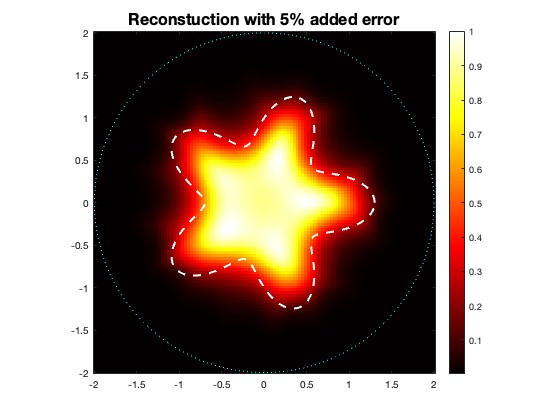}
    \caption{$5\%$ noise}
\end{subfigure}\hfill
\begin{subfigure}[t]{0.33\textwidth}
    \centering
    \includegraphics[width=\linewidth]{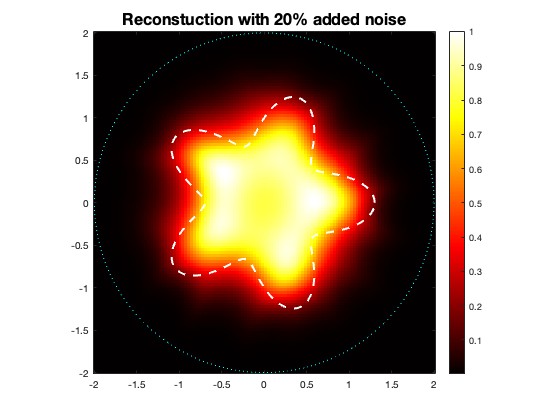}
    \caption{$20\%$ noise}
\end{subfigure}

\caption{Reconstruction of the five-pointed star-shaped scatterer at a fixed wavenumber $\kappa=2$ using the imaging function
$W^{\mathrm{Reg}}(z)$ under increasing levels of noise in the near-field data.
The reconstruction remains robust under noise overall, 
with some degradation of the imaging quality observed for $5\%$ and $20\%$ noise.
We used $40$ sources and receivers.}
\label{fig:star5_noise}
\end{figure}

\begin{figure}[ht]
\centering

\begin{subfigure}[t]{0.33\textwidth}
    \centering
    \includegraphics[width=\linewidth]{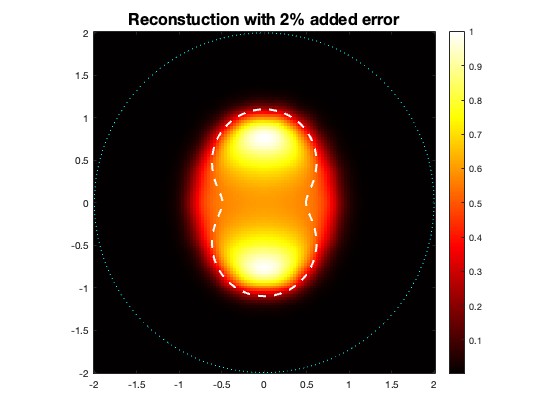}
    \caption{$2\%$ noise}
\end{subfigure}\hfill
\begin{subfigure}[t]{0.33\textwidth}
    \centering
    \includegraphics[width=\linewidth]{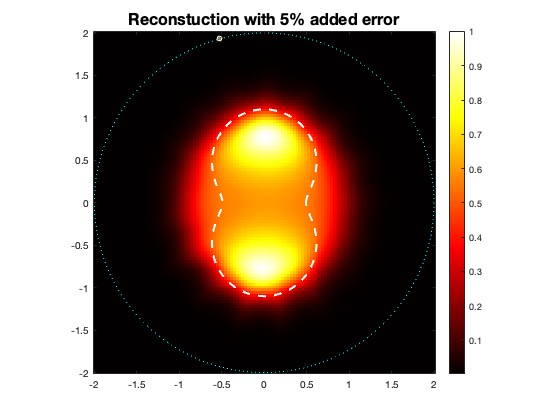}
    \caption{$5\%$ noise}
\end{subfigure}\hfill
\begin{subfigure}[t]{0.33\textwidth}
    \centering
    \includegraphics[width=\linewidth]{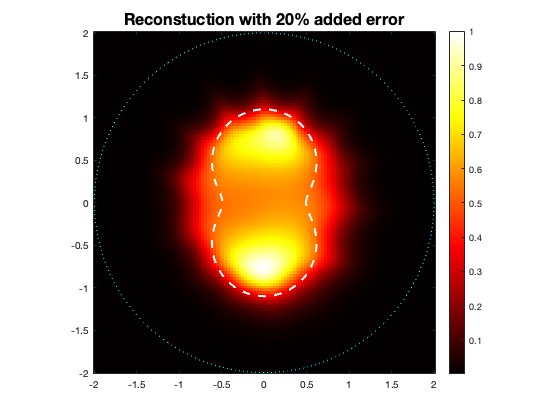}
    \caption{$20\%$ noise}
\end{subfigure}

\caption{Reconstruction of the peanut-shaped scatterer at a fixed wavenumber $\kappa=\pi$ by the imaging function
$W^{\mathrm{Reg}}(z)$ under increasing noise levels in the near-field data.
Here we used $40$ sources and receivers.}
\label{fig:peanut_noise}
\end{figure}

\begin{figure}[h]
\centering

\begin{subfigure}[t]{0.33\textwidth}
    \centering
    \includegraphics[width=\linewidth]{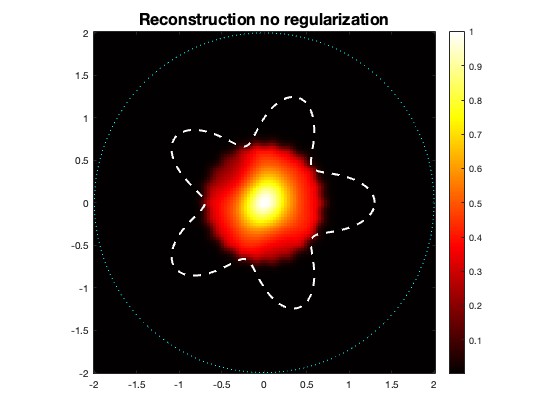}
    \caption{No regularization}
\end{subfigure}\hfill
\begin{subfigure}[t]{0.33\textwidth}
    \centering
    \includegraphics[width=\linewidth]{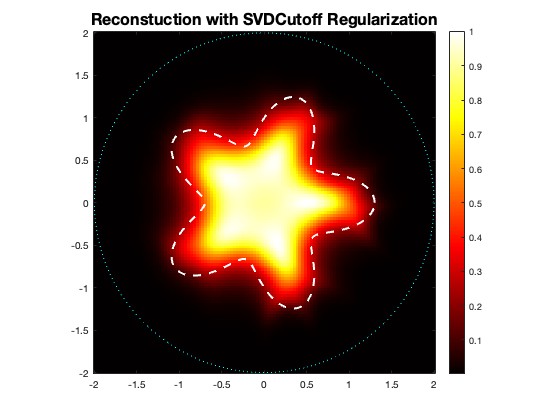}
    \caption{Spectral cutoff}
\end{subfigure}\hfill
\begin{subfigure}[t]{0.33\textwidth}
    \centering
    \includegraphics[width=\linewidth]{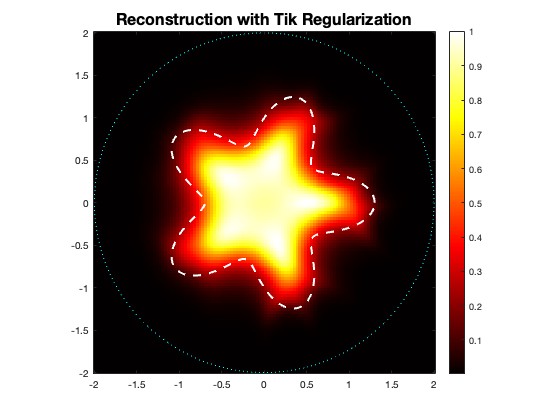}
    \caption{Tikhonov regularization}
\end{subfigure}

\caption{Reconstruction of the five-pointed star-shaped obstacle by the imaging function
$W^{\mathrm{Reg}}(z)$ using different regularization strategies for the near-field
operator. Here $\kappa=2$ and we used $40$ sources and receivers. We also added $2\%$ random relative noise to the near-field data.
The SVD cutoff and Tikhonov regularization results use the regularization parameter
$\alpha=10^{-6}$.}
\label{fig:star_regularization}
\end{figure}

\begin{figure}[ht]
\centering

\begin{subfigure}[t]{0.33\textwidth}
    \centering
    \includegraphics[width=\linewidth]{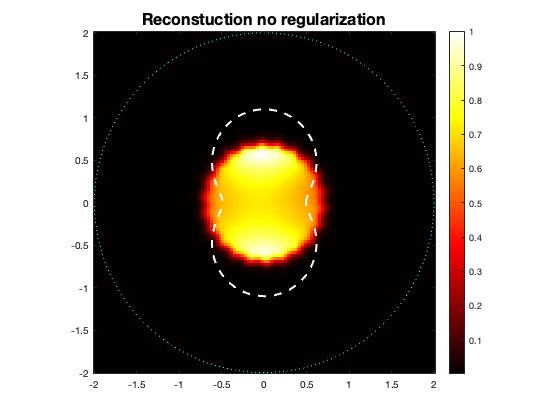}
    \caption{No regularization}
\end{subfigure}\hfill
\begin{subfigure}[t]{0.33\textwidth}
    \centering
    \includegraphics[width=\linewidth]{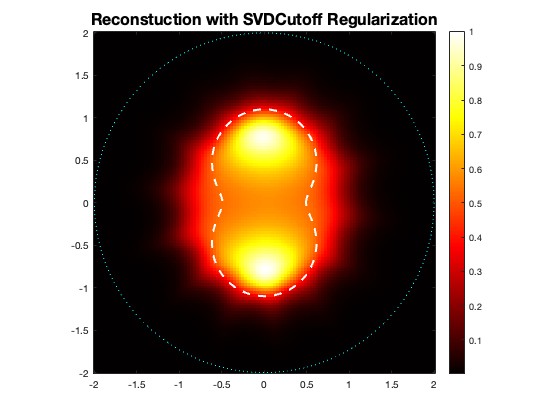}
    \caption{Spectral cutoff}
\end{subfigure}\hfill
\begin{subfigure}[t]{0.33\textwidth}
    \centering
    \includegraphics[width=\linewidth]{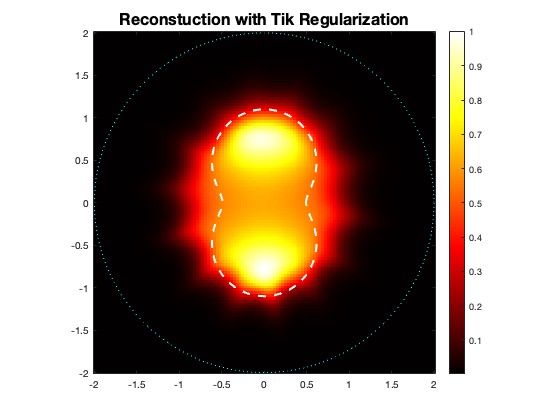}
    \caption{Tikhonov regularization}
\end{subfigure}

\caption{Reconstruction of the peanut-shaped obstacle using different regularization strategies. Here $\kappa=\pi$ and we used $40$ sources and receivers. We also added $15\%$ random relative noise to the near-field data. The SVD cutoff and Tikhonov regularization results use the regularization parameter
$\alpha=10^{-6}$.}
\label{fig:peanut_reg_noise15}
\end{figure}

\subsection{Example 3: Regularization and parameter selection}
For our next set of examples, we wish to test the effect of the reconstruction method with and without regularization. To investigate the effect of regularization, we compare the reconstruction of the five-pointed star-shaped obstacle using the imaging function $W^{\mathrm{Reg}}(z)$ with different regularization strategies in Figure~\ref{fig:star_regularization}. The near-field data are perturbed by $2\%$ relative random noise, with fixed wavenumber of $\kappa=2$ for the five-pointed star and $\kappa=\pi$ for the peanut-shaped obstacle and $40$ sources and receivers. In the case without regularization, we set $\alpha=0$ in the imaging function
$W^{\mathrm{Reg}}(z)$. This result is compared with the reconstructions obtained
using spectral cutoff and Tikhonov regularization with the ad-hoc parameter
$\alpha=10^{-6}$. As expected, the unregularized case is highly unstable under noisy data, producing a distorted imaging function where the shape of the obstacle is no longer identifiable. In contrast, both regularization schemes yield stable and nearly identical reconstructions. A similar behavior is observed for the peanut-shaped obstacle reconstructions, as shown in Figure~\ref{fig:peanut_reg_noise15}. Here, the near-field data are contaminated by $15\%$ relative random noise with fixed wavenumber $\kappa=\pi$ and $40$ sources and receivers. As with the star-shaped obstacle, the unregularized reconstruction of the peanut is significantly affected by the noise, producing a distorted imaging function. In contrast, both the spectral cutoff and Tikhonov regularization schemes with $\alpha=10^{-6}$ provide stable and comparable reconstructions, successfully preserving the main geometric features of the obstacle. There is additionally a more noticeable effect of the noise on the reconstructions.

In our previous examples, we numerically tested the stability of the reconstructions under noisy data and considered two regularization techniques: spectral cutoff and Tikhonov regularization. Nevertheless, the regularization parameter was selected in an ad-hoc manner by fixing $\alpha=10^{-6}$. We now investigate a parameter choice rule motivated by regularization theory and assess its effectiveness in comparison with the fixed empirical choice. In order to provide an analytical strategy for selecting the regularization parameter, a discrepancy principle was introduced in \cite{harris2023rfm}, assuming that the noise level $\delta$ is known. For the Tikhonov regularization, it is given that 
\begin{align}\label{optimal_alpha}
    \alpha_{\mathrm{Tik}}(\delta)&=\frac{1}{4}\delta^{\left(\frac{1}{4}-p\right)}\quad \text{for any}\quad p\in (0,1/4).
\end{align}
Since $\delta$ is not known explicitly in many applications, we approximate it via 
\[
\text{Error}=||\mathbf N-\mathbf N^\top||_2/||\mathbf N||_2
\]
as in \cite{HarrisKleefeld2026}, since we expect $N$ to be symmetric by Lemma 2.4 in \cite{BourgeoisHazard2020}.
We note that this is a good approximation of the noise level since a reciprocity relationship was shown for the scattered field $u^s$ generated by a point source $u^i=\mathbb G(\cdot \, ,y)$ in \cite{DongLi2024Uniqueness}, implying $\mathbf N$ is a symmetric matrix. Figure~\ref{fig:disk_alpha_comparison} compares the two Tikhonov regularization parameter choices for the disk-shaped obstacle under $20\%$ noise. The noise-dependent choice produces a noticeably more robust reconstruction to the ad-hoc choice of $\alpha$. Similarly,
in Figure \ref{fig:peanut_alpha_comparison}, we provide provide the reconstruction of the peanut-shaped scatterer with $\delta=0.05$ for $\alpha=10^{-6}$ chosen ad-hoc and optimally via \eqref{optimal_alpha} for the Tikhonov filter function. Here we see that the reconstruction using the optimal regularization parameter gives an image that's more robust towards noise. In addition, Figure \ref{fig:star5_alpha_comparison} provides a reconstruction of the star-shaped scatterer with $\delta=0.10$, where the same comparison between the fixed ad-hoc parameter choice and the noise-dependent parameter selection is performed. Similarly, we observe that the parameter choice given by \eqref{optimal_alpha} improves the stability of the reconstruction in the presence of noisy data.

\begin{figure}[ht]
\centering

\begin{subfigure}[t]{0.48\textwidth}
    \centering
    \includegraphics[width=\linewidth]{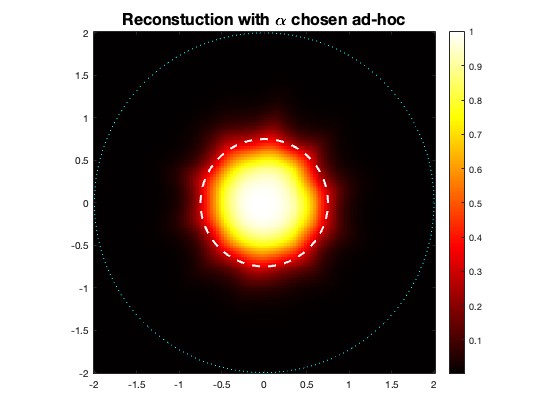}
    \caption{Tikhonov regularization with fixed $\alpha=10^{-6}$}
\end{subfigure}\hfill
\begin{subfigure}[t]{0.48\textwidth}
    \centering
    \includegraphics[width=\linewidth]{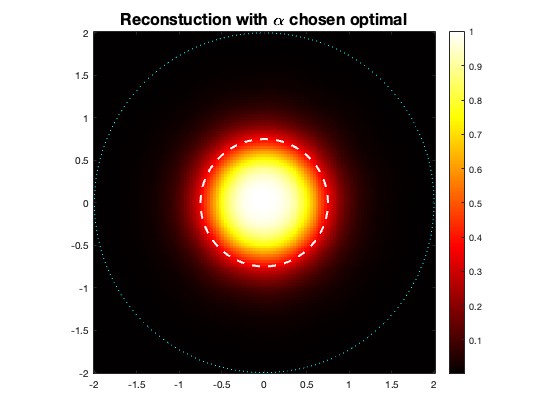}
    \caption{Tikhonov regularization with noise-dependent $\alpha$}
\end{subfigure}

\caption{Comparison of Tikhonov regularization parameter choices for the
reconstruction of the disk-shaped obstacle with radius $=0.75$. The left
figure uses the fixed ad-hoc parameter $\alpha=10^{-6}$, while the right
figure uses the parameter choice
$\alpha=0.25\cdot\text{Error}^{1/8}$. Here, $\kappa=2$, $40$ sources
and receivers are used, and $\delta=20\%$ random relative noise is added
to the near-field data.}
\label{fig:disk_alpha_comparison}
\end{figure}

\begin{figure}[ht]
\centering

\begin{subfigure}[t]{0.48\textwidth}
    \centering
    \includegraphics[width=\linewidth]{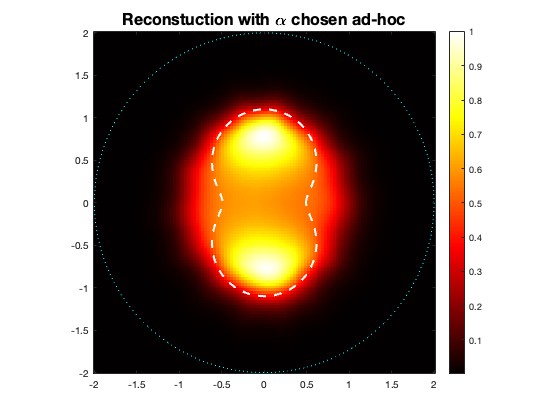}
    \caption{Tikhonov regularization with fixed $\alpha=10^{-6}$}
\end{subfigure}\hfill
\begin{subfigure}[t]{0.48\textwidth}
    \centering
    \includegraphics[width=\linewidth]{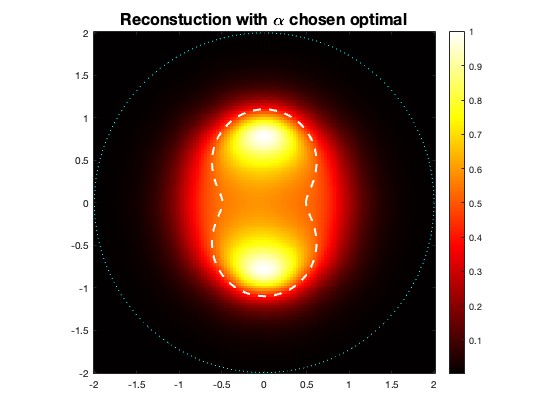}
    \caption{Tikhonov regularization with noise-dependent $\alpha$}
\end{subfigure}

\caption{Comparison of Tikhonov regularization parameter choices for the reconstruction of the peanut-shaped obstacle. The left figure uses the fixed ad-hoc parameter $\alpha=10^{-6}$, while the right figure uses the parameter choice $\alpha=0.25\cdot \text{Error}^{1/8}$. Here
 $\kappa=\pi$, $40$ sources and receivers are used, and $\delta=5\%$ random relative noise is added to the near-field data.}
\label{fig:peanut_alpha_comparison}
\end{figure}

\begin{figure}[ht]
\centering

\begin{subfigure}[t]{0.48\textwidth}
    \centering
    \includegraphics[width=\linewidth]{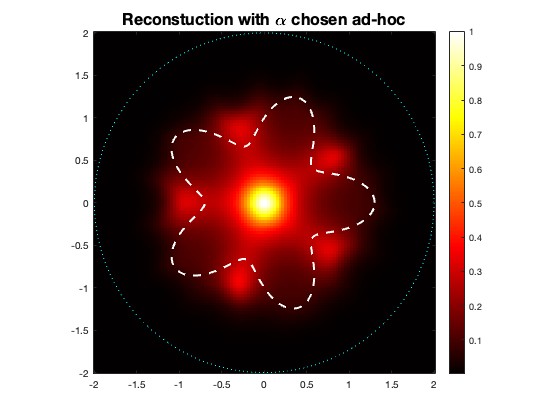}
    \caption{Tikhonov regularization with fixed $\alpha=10^{-6}$}
\end{subfigure}\hfill
\begin{subfigure}[t]{0.48\textwidth}
    \centering
    \includegraphics[width=\linewidth]{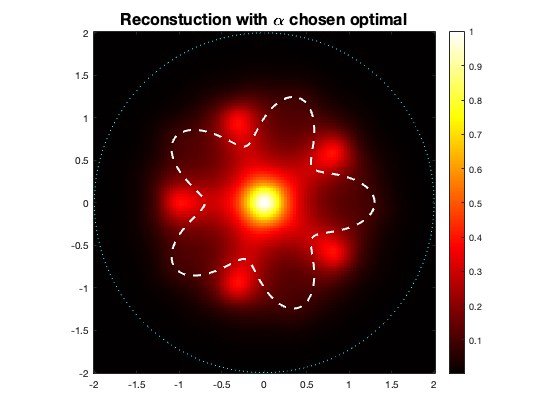}
    \caption{Tikhonov regularization with noise-dependent $\alpha$}
\end{subfigure}

\caption{Comparison of Tikhonov regularization parameter choices for the reconstruction of the five-pointed star-shaped obstacle. The left figure uses the fixed ad-hoc parameter $\alpha=10^{-6}$, while the right figure uses the noise-dependent parameter choice $
\alpha=0.25\cdot \mathrm{Error}^{1/8}$. Here $\kappa=\pi$, $40$ sources and receivers are used, and $\delta=10\%$ random relative noise is added to the near-field data.}
\label{fig:star5_alpha_comparison}
\end{figure}

\subsection{Example 4: Reconstructions using only scattered field measurements}
In this section, we consider the case when one only has scattered field data generated by a point source not measurements from a dipole. Note that the theoretical results require the multi-static data given by \eqref{NF_dataset} to ensure the unique reconstruction.  However, in practical applications, obtaining normal derivative measurements
and generating dipole incident fields is generally infeasible. By comparison, direct measurements of the scattered field are considerably more natural and accessible than recovering the full Cauchy data on the measurement curve. Here we wish to measure how well our reconstruction method works in the case where only the scattered field data is available.

Assuming that one only has the measured field $u^s$, we can express the field $u^s$ in terms of two auxiliary functions $u_{\mathrm H}^s$ and $u_{\mathrm M}^s$ that are defined by the expressions
\begin{equation}\label{aux_components}
    u_{\mathrm H}^s=-\frac{1}{2\kappa^2}(\Delta u^s-\kappa^2u^s)\quad \text{and}\quad u_{\mathrm M}^s=\frac{1}{2\kappa^2}(\Delta u^s+\kappa^2u^s).
\end{equation}
This would imply
\[
u^s=u_{\mathrm H}^s+u_{\mathrm M}^s,
\]
just as in \cite{HarrisLiOzochiawaeze2026} and other works. Here, as in other works, we call $u_{\mathrm H}^s$ the propagative part and $u_{\mathrm M}^s$ the evanescent part. Using the splitting of the biharmonic Helmholtz operator
\[
\Delta^2-\kappa^4=(\Delta-\kappa^2)(\Delta+\kappa^2)=(\Delta+\kappa^2)(\Delta-\kappa^2),
\]
one obtains that $u_{\mathrm H}^s$ and $u_{\mathrm M}^s$ satisfy the Helmholtz and modified Helmholtz equations, respectively, outside of the scatterer $D$, i.e., the scattering problem \eqref{eqnbcs}-\eqref{SRCs} decouples to 
\begin{align*}
    \Delta u_{\mathrm H}^s+\kappa^2 u_{\mathrm H}^s=0\quad \text{and}\quad \Delta u_{\mathrm M}^s-\kappa^2 u_{\mathrm M}^s=0\quad \text{in }\mathbb R^2\setminus\overline{D}\\
    u_{\mathrm H}^s+u_{\mathrm M}^s=-u^i\quad \partial_\nu(u_{\mathrm H}^s+u_{\mathrm M}^s)=-\partial_\nu u^i\quad \text{on }\partial D
\end{align*}
along with the radiation conditions
\begin{align}\label{decouple_3}
    \lim_{r\to\infty}\sqrt r(\partial_r u_{\mathrm H}^s-\mathrm i\kappa u_{\mathrm H}^s)=0\quad \text{and}\quad \lim_{r\to\infty}\sqrt r(\partial_r u_{\mathrm M}^s-\mathrm i\kappa u_{\mathrm M}^s)=0\quad \text{for }r=|x|.
\end{align}
Moreover, it is well-known that $u_{\mathrm M}^s$ and $\partial_r u_{\mathrm M}^s$ decay exponentially (see, e.g.,\cite{BourgeoisHazard2020}, \cite{HarrisLeeLi2025}). 

Assuming that only scattered field measurements generated by point source excitations are available, we define the reduced near-field matrix
\[
\mathbf N_{\mathrm{scat}}
=
\left(
u_N^s(X_k,Y_\ell)
\right)_{k,\ell=1}^{40}
\in\mathbb C^{40\times40}.
\]
In terms of the full near-field matrix \eqref{NFmatrix}, this corresponds to
retaining only the scattered field block generated by point sources. Hence,
the full near-field data matrix may be viewed schematically as
\[
\mathbf N=
\begin{pmatrix}
\mathbf N_{\mathrm{scat}} & *\\
* & *
\end{pmatrix},
\]
where the remaining blocks contain measurements involving dipole excitations
and/or normal derivative data. Due to the decomposition $
u^s=u_{\mathrm H}^s+u_{\mathrm M}^s,
$
and the exponential decay of the modified Helmholtz component,
\[
u_{\mathrm M}^s(x)
=
\mathcal O\left(\mathrm e^{-\kappa r}/\sqrt{r}\right),
\qquad r\to\infty,
\]
the scattered field measured on $\mathcal C=\partial B_R$ can be used to approximate the
Helmholtz component. Therefore, for sufficiently large $R$, we expect the reduced matrix
$\mathbf N_{\mathrm{scat}}$ to contain enough information for reconstructing the clamped obstacles. 

We now define the far-field transformation used in this setting. We first consider the auxiliary exterior Helmholtz problem
\[
\Delta w+\kappa^2 w=0
\qquad\text{in }\mathbb R^2\setminus \overline{\text{Int}(\mathcal C}),
\]
where $w$ satisfies the Sommerfeld radiation condition. Let 
$
f=w|_{\mathcal C}$
denote the Dirichlet trace of the radiating solution on the measurement curve. The associated Helmholtz Dirichlet-to-far-field transform is defined by
\[
\mathcal Q_{\mathrm{acou}}:
H^{1/2}(\mathcal C)\rightarrow L^2(\mathbb S^1)\quad \text{is given by}\quad
\mathcal Q_{\mathrm{acou}}f=w^\infty,
\]
where $v^\infty$ denotes the far-field pattern of $v$. Using $\mathcal C=\partial B_R$ and separation of variables, the transformation
$\mathcal Q_{\mathrm{acou}}$ can be written as
\[
(\mathcal Q_{\mathrm{acou}}f)(\theta)
=
\int_0^{2\pi}
Q_{\mathrm{acou}}^{\mathrm{func}}(\theta,\varphi)f(\varphi)\,d\varphi,\quad \text{where}\quad Q_{\mathrm{acou}}^{\mathrm{func}}(\theta,\varphi)
=
\frac{2}{\pi \mathrm i}
\sum_{|m|=0}^{\infty}
\frac{\mathrm e^{\mathrm im(\theta-\varphi-\pi/2)}}
{H_m^{(1)}(\kappa R)},
\]
as shown in \cite{HarrisNyugens2022}. With this, defining 
\[
\mathbf Q_{\mathrm{acou}}=[Q_{\mathrm{acou}}^{\text{func}}(\theta_k,\theta_{\ell})]_{k,\ell=1}^{40}\quad \text{where}\quad Q_{\mathrm{acou}}^{\text{func}}(\theta,\varphi)=\frac{2}{\pi \mathrm i}
\sum_{|m|=0}^{10}
\frac{\mathrm e^{\mathrm im(\theta-\varphi-\pi/2)}}
{H_m^{(1)}(\kappa R)},
\]
we expect that our reconstruction method will work if we use the eigenvectors and eigenvalues of the matrix $\mathbf Q_{\mathrm{acou}}\mathbf N_{\mathrm{scat}}\mathbf Q_{\mathrm{acou}}^\top \mathbf R\in \mathbb C^{40\times 40}$ in our imaging functional $W^{\mathrm{Reg}}(z)$. A similar idea was used in \cite{BourgeoisRecoquillay2020} when the linear sampling method was applied to recovering the clamped obstacle with near-field data. 

We start with Figure~\ref{fig:disk_scat_comparison}, which compares reconstructions of the
disk-shaped obstacle using the full near-field matrix
$\mathbf N^\delta$ and the scattered-field matrix
$\mathbf N_{\mathrm{scat}}^\delta$. Both use the noise-dependent parameter
$\alpha=0.25\cdot\mathrm{Error}^{1/8}$. Here, $\kappa=2$, with the disk of radius $=0.75$, and
$\delta=2\%$ random relative noise is added to the data. While the full-data reconstruction provides a more precise recovery of the boundary, the scattered-field reconstruction remains effective, with only an additional artifact in the form of a circular halo surrounding the obstacle.
\begin{figure}[ht]
\centering

\begin{subfigure}[t]{0.48\textwidth}
    \centering
    \includegraphics[width=\linewidth]{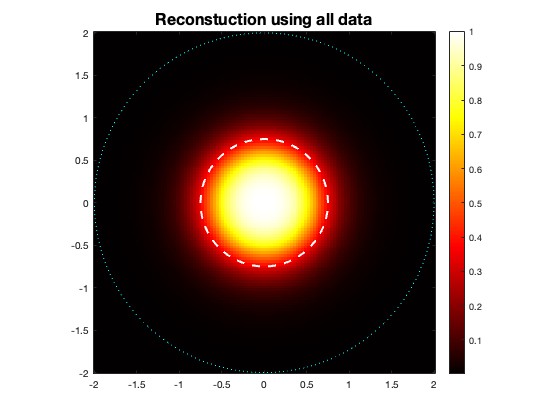}
    \caption{Full near-field matrix $\mathbf N^\delta$}
\end{subfigure}\hfill
\begin{subfigure}[t]{0.48\textwidth}
    \centering
    \includegraphics[width=\linewidth]{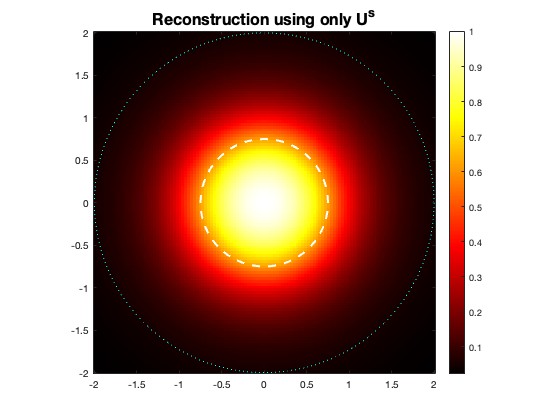}
    \caption{Scattered-field matrix $\mathbf N_{\mathrm{scat}}^\delta$}
\end{subfigure}

\caption{Comparison of reconstructions of the disk-shaped obstacle using
the full near-field matrix and the scattered-field matrix.}
\label{fig:disk_scat_comparison}
\end{figure}
Figure \ref{fig:peanut_scat_comparison} compares the reconstructions of our peanut-shaped obstacle  obtained
using the full near-field matrix $\mathbf N^\delta$ with the reconstruction
obtained using only the scattered-field matrix
$\mathbf N_{\mathrm{scat}}^\delta$. Both reconstructions use the
noise-dependent regularization parameter
$\alpha=0.25\cdot\mathrm{Error}^{1/8}$. Here, we consider the peanut-shaped
obstacle with $\kappa=\pi$, measurement curve
$\mathcal C=\partial B_R$ with $R=2$, and $\delta=10\%$ random relative noise added to
the data. The reconstruction using the scattered-field-only data successfully captures
the main features of the obstacle, demonstrating that the far-field
transformation based on $\mathcal Q_{\mathrm{acou}}$ can extract useful
information from reduced near-field measurements. As expected, the use of the
full near-field matrix provides better resolution and results in a sharper
reconstruction of the boundary. Nevertheless, the reconstruction obtained from
$\mathbf N_{\mathrm{scat}}^\delta$ remains stable.  

Figure
\ref{fig:star5_scat_comparison} demonstrate the effect of replacing the full
near-field data by scattered-field-only measurements for the five-pointed star obstacle. Again we use the
noise-dependent regularization parameter
$\alpha=0.25\cdot\mathrm{Error}^{1/8}$. Moreover we select the parameters $\kappa=\pi$, a measurement curve
$\mathcal C=\partial B_R$ with $R=2$, and $\delta=2\%$ random relative noise is added to the data. For the five-pointed star-shaped
obstacle, the overall location and geometry are still identifiable; however, the reconstruction contains additional artifacts that partly obscure the boundary.  This indicates that while the scattered-field-only formulation retains the essential information required for reconstruction, the absence of derivative and dipole measurements can affect the reconstruction quality for more geometrically complex obstacles.

\begin{figure}[ht]
\centering

\begin{subfigure}[t]{0.48\textwidth}
    \centering
    \includegraphics[width=\linewidth]{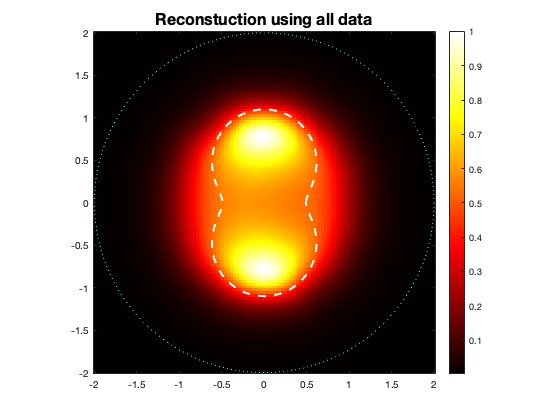}
    \caption{Reconstruction using the full near-field data $\mathbf N^\delta$}
\end{subfigure}\hfill
\begin{subfigure}[t]{0.48\textwidth}
    \centering
    \includegraphics[width=\linewidth]{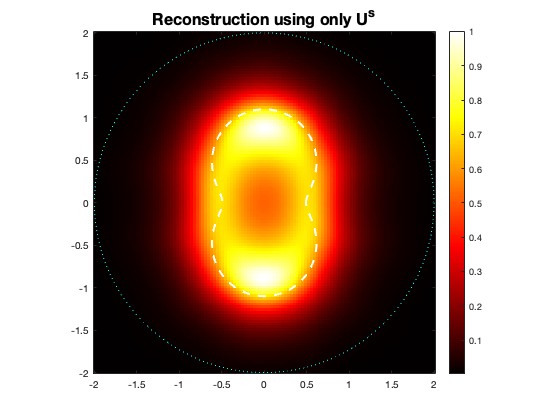}
    \caption{Reconstruction using only scattered field data 
    $\mathbf N_{\mathrm{scat}}^\delta$}
\end{subfigure}

\caption{Comparison of the factorization method using full multi-static data and
scattered-field-only data for the reconstruction of the peanut-shaped obstacle. Both reconstructions use the noise-dependent regularization parameter $\alpha=0.25\cdot\mathrm{Error}^{1/8}$.
Here $\kappa=\pi$, $40$ sources and receivers are used, and $\delta=10\%$
random relative noise is added to the near-field data.}
\label{fig:peanut_scat_comparison}
\end{figure}

\begin{figure}[ht]
\centering

\begin{subfigure}[t]{0.48\textwidth}
    \centering
    \includegraphics[width=\linewidth]{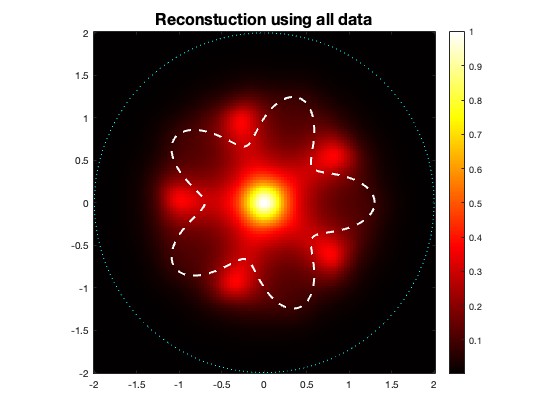}
    \caption{Reconstruction using full near-field data $\mathbf N^\delta$}
\end{subfigure}\hfill
\begin{subfigure}[t]{0.48\textwidth}
    \centering
    \includegraphics[width=\linewidth]{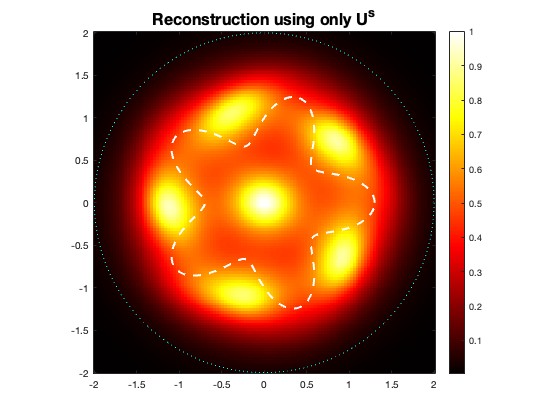}
    \caption{Reconstruction using scattered-field-only data
    $\mathbf N_{\mathrm{scat}}^\delta$}
\end{subfigure}

\caption{Comparison of reconstructions for the five-pointed star-shaped
obstacle using the full near-field matrix and the scattered-field-only
matrix. Both reconstructions use the noise-dependent regularization parameter
$\alpha=0.25\cdot\mathrm{Error}^{1/8}$. Here,
$\kappa=\pi$, and $2\%$ random relative
noise is added to the near-field data.}
\label{fig:star5_scat_comparison}
\end{figure}

\section{Conclusion} \label{conclude_sect}
In this work, we developed a rigorous factorization method framework for reconstructing a clamped obstacle in a Kirchhoff-Love thin plate from near-field measurements via a far-field transformation. Building on the work in \cite{BourgeoisRecoquillay2020}, we presented a theoretical factorization of the near-field operator and proposed a numerical algorithm to reconstruct the shape and location of the clamped obstacle. Specifically, the far-field transform $\mathcal{Q}$ played a crucial role by symmetrizing the near-field operator and cleanly connecting the scattering operator to the far-field operator, for which the analytical machinery of the factorization method justified in \cite{Zhu2026} applies. Numerical experiments are performed to demonstrate the effectiveness of our method. 

While dipole measurements were ultimately instrumental in securing the theoretical framework and addressing the complexities arising from the uniqueness result shown in \cite{BourgeoisRecoquillay2020}, our numerical experiments demonstrate that point source measurements alone are practically sufficient to achieve accurate reconstructions. This gap between the current analytical data requirements and empirical performance highlights an exciting direction for future work: analytically weakening the data assumptions to rigorously justify the sufficiency of point source data alone, as well as extending these techniques to other boundary conditions such as free plate obstacles. Note that the factorization method developed here requires two additional assumptions on the wavenumber $\kappa$: namely, that $\kappa$ is neither a clamped transmission eigenvalue nor an interior clamped eigenvalue. The latter means that $\kappa^4$ is not a Dirichlet eigenvalue of the bilaplacian operator $\Delta^2$ in the clamped domain $D$. In contrast, the uniqueness result established in \cite{BourgeoisRecoquillay2020} does not require any such spectral assumptions on $\kappa$. These additional restrictions therefore arise from the present factorization framework rather than from the underlying uniqueness of the inverse problem. Removing at least one of these spectral assumptions for reconstruction is another interesting direction for future work.

\appendix

\section{Large-Order Asymptotics of the Fourier Coefficients}
In this appendix, we derive the large-order asymptotics of the Fourier
coefficients $A_m$ and $B_m$ appearing in the representation of the
far-field transform $\mathcal{Q}$. These estimates are subsequently used to
establish the convergence rate of the truncated far-field transform
$\mathcal Q_M$.
\label{app:asymptotics}
\begin{Lemma}\label{lem:coeff_asymptotics}
Let $A_m$ and $B_m$ be the Fourier coefficients defined by \eqref{first_coeff} and
\eqref{second_coeff}, respectively. Then, as $m\to\infty$, we obtain
\[
|A_m|^2+|B_m|^2
\sim
\frac{\pi m}{2\kappa^2}
\left(\frac{\mathrm e\kappa R}{2m}\right)^{2m}. \quad \text{In particular,
$
|A_m|^2+|B_m|^2
=
\mathcal O\left(
m
\left(\frac{\mathrm e\kappa R}{2m}\right)^{2m}
\right).
$}
\]
\end{Lemma}

\begin{proof}
We use the standard large-order asymptotic formula
\[
-\mathrm{i}H_m^{(1)}(t)
\sim
\sqrt{\frac{2}{\pi m}}
\left(\frac{\mathrm et}{2m}\right)^{-m},
\qquad m\to\infty,
\]
for fixed $t\neq0$ referenced in \cite{NIST:DLMF}. Moreover, we note the recurrence relation
\begin{align*}
(H_m^{(1)})'(t)=
\frac12\left(
H_{m-1}^{(1)}(t)-H_{m+1}^{(1)}(t)
\right)= \frac{1}{2}H_m^{(1)}(t)\left(\frac{H_{m-1}^{(1)}(t)}{H_m^{(1)}(t)}-\frac{H_{m+1}^{(1)}(t)}{H_m^{(1)}(t)}\right)
\end{align*}
holds. Then for a fixed $t$, we see that the large-order asymptotic formula gives that 
\begin{align*}
    \frac{H_{m+1}^{(1)}(t)}{H_m^{(1)}(t)}&\sim \sqrt{\frac{m}{m+1}}\left(\frac{\mathrm et}{2(m+1)}\right)^{-(m+1)}\left(\frac{\mathrm et}{2m}\right)^{m}\\
    &=\sqrt{\frac{m}{m+1}}\left(\frac{2(m+1)}{\mathrm et}\right)\left(\frac{2(m+1)}{\mathrm et}\cdot \frac{\mathrm et}{2m}\right)^m\\
    &=\sqrt{\frac{m}{m+1}}\left(\frac{2(m+1)}{\mathrm et}\right)\left(1+\frac{1}{m}\right)^m
    \end{align*}
    which simplifies to give that 
    \begin{align*}
    \frac{H_{m+1}^{(1)}(t)}{H_m^{(1)}(t)}&\sim   \frac{2(m+1)}{\mathrm{e}t}\cdot \mathrm{e} =\frac{2(m+1)}{t}\sim \frac{2m}{t} \quad \text{as $m\to\infty$.}
\end{align*}
Thus,
\[
\frac{H_{m+1}^{(1)}(t)}{H_m^{(1)}(t)}
\sim \frac{2m}{t}\quad \text{as $m\to\infty$.}
\]
Similarly, by replacing $m$ with $m-1$ in the large-order asymptotic formula,
\[
\frac{H_{m-1}^{(1)}(t)}{H_m^{(1)}(t)}
\sim \frac{t}{2m} \quad \text{as $m\to\infty$.}
\]
Therefore, substituting the asymptotics in the recurrence relation yields
\[
-\mathrm i(H_m^{(1)})'(t)
\sim
\frac{m}{t}
\sqrt{\frac{2}{\pi m}}
\left(\frac{\mathrm et}{2m}\right)^{-m}\quad \text{as $m\to\infty$.}
\]
Consequently,
\[
H_m^{(1)}(\kappa R)
\sim
\mathrm i\sqrt{\frac{2}{\pi m}}
\left(\frac{\mathrm{e}\kappa R}{2m}\right)^{-m}\quad \text{and}\quad (H_m^{(1)})'(\kappa R)
\sim
\mathrm i\frac{m}{\kappa R}
\sqrt{\frac{2}{\pi m}}
\left(\frac{\text{e}\kappa R}{2m}\right)^{-m} \quad \text{as $m\to\infty$.}
\]
For the imaginary argument, we similarly have
\[
H_m^{(1)}(\mathrm{i}\kappa R)
\sim
\mathrm i^{m+1}
\sqrt{\frac{2}{\pi m}}
\left(\frac{\mathrm{e}\kappa R}{2m}\right)^{-m}\quad \text{and}\quad (H_m^{(1)})'(\mathrm{i}\kappa R)
\sim
\mathrm{i}^{m+1}\frac{m}{\text{i}\kappa R}
\sqrt{\frac{2}{\pi m}}
\left(\frac{\mathrm e\kappa R}{2m}\right)^{-m} \quad \text{as $m\to\infty$.}
\]
Thus, the denominator appearing in \eqref{first_coeff}--\eqref{second_coeff}
satisfies
\begin{align*}
&\mathrm i\kappa H_m^{(1)}(\kappa R)(H_m^{(1)})'(\mathrm i\kappa R)
-\kappa(H_m^{(1)})'(\kappa R)H_m^{(1)}(\mathrm i\kappa R)
\\
&\quad\sim
\frac{m}{R}\mathrm i^{m+1}
\frac{2}{\pi m}
\left(\frac{\mathrm e\kappa R}{2m}\right)^{-2m}
-
\frac{m}{R}\mathrm i^{m+2}
\frac{2}{\pi m}
\left(\frac{\mathrm e\kappa R}{2m}\right)^{-2m}
\\
&\quad=
\frac{2}{\pi R}\mathrm i^m(1+\mathrm i)
\left(\frac{\mathrm e\kappa R}{2m}\right)^{-2m},
\end{align*}
where we have used $\mathrm i^{m+2}=-\mathrm i^m$.
For $A_m$, its numerator satisfies
$$
2(H_m^{(1)})'(\mathrm i\kappa R)
\sim
2\mathrm i^m\frac{m}{\mathrm i\kappa R}
\sqrt{\frac{2}{\pi m}}
\left(\frac{\mathrm e\kappa R}{2m}\right)^{-m}
\quad \text{as $m\to\infty$.}
$$
Thus, substitution into \eqref{first_coeff} gives
$$
A_m
\sim
\frac{1-\mathrm i}{2}
\sqrt{\frac{2\pi m}{\kappa^2}}
\left(\frac{\mathrm e\kappa R}{2m}\right)^m,
\quad \text{and hence} \quad
|A_m|^2
\sim
\frac{\pi m}{2\kappa^2}
\left(\frac{\mathrm e\kappa R}{2m}\right)^{2m}
\quad \text{as $m\to\infty$.}
$$
Similarly, for $B_m$, its numerator satisfies

$$
2\mathrm i H_m^{(1)}(\mathrm i\kappa R)
\sim
2\mathrm i^{m+2}
\sqrt{\frac{2}{\pi m}}
\left(\frac{\mathrm e\kappa R}{2m}\right)^{-m}
\quad \text{as $m\to\infty$.}
$$
Thus, substitution into \eqref{second_coeff} gives
$$
B_m
\sim
\frac{\mathrm i-1}{2}
\sqrt{\frac{2\pi R^2}{m}}
\left(\frac{\mathrm e\kappa R}{2m}\right)^m,
\quad \text{and hence} \quad
|B_m|^2
\sim
\frac{\pi R^2}{2m}
\left(\frac{\mathrm e\kappa R}{2m}\right)^{2m}
\quad \text{as $m\to\infty$.}
$$
Therefore,
\begin{align*}
|A_m|^2+|B_m|^2
&\sim
\left(
\frac{\pi m}{2\kappa^2}
+
\frac{\pi R^2}{2m}
\right)
\left(\frac{\mathrm e\kappa R}{2m}\right)^{2m}
=
\frac{\pi m}{2\kappa^2}
\left(
1+\frac{\kappa^2R^2}{m^2}
\right)
\left(\frac{\mathrm e\kappa R}{2m}\right)^{2m}
\quad \text{as $m\to\infty$.}
\end{align*}
Since $R$ and $\kappa$ are fixed, we note that
$
1+\frac{\kappa^2R^2}{m^2}\longrightarrow 1
\text{ as }m\to\infty,
$ which proves the result.
\end{proof}

\noindent{\bf Acknowledgments:} The research of authors I. Harris and G. Ozochiawaeze is partially supported by the NSF DMS Grants 2509722 and 2208256. \\


\end{document}